\documentclass[a4paper]{article}

\usepackage{algorithmic}
\usepackage{amsmath}
\usepackage{amssymb}
\usepackage{amsthm}
\usepackage[french,english]{babel}
\usepackage[footnotesize,bf]{caption}
\usepackage{cite}
\usepackage{fixltx2e}
\usepackage[T1]{fontenc}
\usepackage{graphicx}
\usepackage{hyperref}
\usepackage{IEEEtrantools}
\usepackage[utf8]{inputenc}
\usepackage{perpage}

\title{Therapeutic target discovery using Boolean network attractors: avoiding pathological phenotypes}
\author{\foreignlanguage{french}{Arnaud Poret}\textsuperscript{1,*}, \foreignlanguage{french}{Jean-Pierre Boissel}\textsuperscript{2}}
\date{29 November 2013; 8 October 2014; 7 May 2015}

\MakePerPage{footnote}
\algsetup{indent=2em}
\algsetup{linenodelimiter=}
\newtheorem{theorem}{Theorem}

\begin{document}

\maketitle

\vfill

\noindent Copyright 2013-2015 \foreignlanguage{french}{Arnaud Poret}, \foreignlanguage{french}{Jean-Pierre Boissel}.

\bigskip

\noindent This document is licensed under the Creative Commons Attribution-NonCommercial-ShareAlike 4.0 International License. To view a copy of this license, visit \url{http://creativecommons.org/licenses/by-nc-sa/4.0/}.

\bigskip

\noindent \textsuperscript{1} \texttt{arnaud.poret@gmail.com}\\
UMR CNRS 5558 Biometry and Evolutionary Biology Laboratory\\
\foreignlanguage{french}{Villeurbanne}\\
\foreignlanguage{french}{France}\\
\url{https://lbbe.univ-lyon1.fr/}

\bigskip

\noindent \textsuperscript{2} \texttt{jean-pierre.boissel@novadiscovery.com}\\
Novadiscovery\\
\foreignlanguage{french}{Lyon}\\
\foreignlanguage{french}{France}\\
\url{http://www.novadiscovery.com}

\bigskip

\noindent \textsuperscript{*} Corresponding author.

\vfill

\newpage

\begin{abstract}
Target identification, one of the steps of drug discovery, aims at identifying biomolecules whose function should be therapeutically altered in order to cure the considered pathology. This work proposes an algorithm for \textit{in silico} target identification using Boolean network attractors. It assumes that attractors of dynamical systems, such as Boolean networks, correspond to phenotypes produced by the modeled biological system. Under this assumption, and given a Boolean network modeling a pathophysiology, the algorithm identifies target combinations able to remove attractors associated with pathological phenotypes. It is tested on a Boolean model of the mammalian cell cycle bearing a constitutive inactivation of the retinoblastoma protein, as seen in cancers, and its applications are illustrated on a Boolean model of Fanconi anemia. The results show that the algorithm returns target combinations able to remove attractors associated with pathological phenotypes and then succeeds in performing the proposed \textit{in silico} target identification. However, as with any \textit{in silico} evidence, there is a bridge to cross between theory and practice, thus requiring it to be used in combination with wet lab experiments. Nevertheless, it is expected that the algorithm is of interest for target identification, notably by exploiting the inexpensiveness and predictive power of computational approaches to optimize the efficiency of costly wet lab experiments.
\end{abstract}

\tableofcontents

\newpage

\section{Introduction}
Drug discovery, as its name indicates, aims at discovering new drugs against diseases. This process can be segmented into three steps: i) disease model provision, where experimental models are developed, ii) target identification, where therapeutic targets are proposed, and iii) target validation, where the proposed therapeutic targets are assessed. This work focuses on the second step of drug discovery: target identification \cite{knowles2003target,lindsay2003target}.

Given an organism suffering from a disease, target identification aims at finding where to act among its multitude of biomolecules in order to alleviate, or ultimately cure, the physiological consequences of the disease. These biomolecules on which perturbations should be applied are called targets and are targeted by drugs \cite{imming2006drugs}. This raises two questions: which target should be therapeutically perturbed and what type of perturbation should be applied on it. Broadly, the functional perturbation of a target by a drug can be either activating or inactivating, regardless the way the drug achieves it.

One solution is to test all, or at least a large number of, biomolecules for activation and inactivation. Knowing that targeting several biomolecules is potentially more effective \cite{anighoro2014polypharmacology,zimmermann2007multi}, the number of possibilities is consequently huge. This rather brute-force screening can be refined with knowledge about the pathophysiology of interest by identifying potential targets based on the role they play in it \cite{gibbs2000mechanism}. Even with this knowledge, experimentally assessing the selected potential targets through wet lab experiments is far from straightforward since such experiments are costly in time and resources \cite{kaitin2010deconstructing}. Fortunately, owing to their integrative power and low cost compared to wet lab experiments, \textit{in silico} approaches appear as valuable tools in improving the efficiency of target identification \cite{tang2014network,mak2013anti,sarker2013silico,yao2012silico,bahadduri2010targeting,ma2010silico,vujasinovic2010silico,chandra2009computational,saidani2009potential,nielsen2005using,duckworth2002insilico,noble1999biological}, as demonstrated through several works using various computational methods \cite{carels2015computational,frangou2014molecular,li2014computer,nicklas2014silico,rao2013identification,ravindranath2013silico,yang2013silico,iadevaia2010identification,ozbayraktar2010drug,koborova2009silico,dasika2006computational}.

However, the stumbling block of \textit{in silico} approaches is that they are built from the available knowledge: not all is known about everything. Nevertheless, an impressive and ever increasing amount of biological knowledge is already available in the scientific literature, databases and knowledge bases such as, to name a few, DrugBank \cite{wishart2008drugbank}, KEGG \cite{kanehisa2000kegg}, PharmGKB \cite{whirl2012pharmacogenomics}, Reactome \cite{croft2011reactome} and TTD \cite{chen2002ttd}. In addition to the difficulty of integrating an increasing body of knowledge comes the inherent complexity of biological systems themselves \cite{kitano2002systems}: this is where computational tools can help owing to their integrative power \cite{boissel2009modeling,boissel2008modelling,kitano2002computational}. This interplay between wet lab and computational biology is synergistic rather than competitive \cite{di2006vivo}. Since wet lab experiments produce factual results, they can be considered as trustworthy sources of knowledge. Once these factual pieces of knowledge are obtained, computational tools can help to integrate them and infer new ones. This computationally obtained knowledge can be subsequently used to direct further wet lab experiments, thus mutually potentiating the whole.

The goal of this work is to propose a computational methodology implemented in an algorithm for \textit{in silico} therapeutic target discovery using Boolean network attractors. It assumes that Boolean network attractors correspond to phenotypes produced by the modeled biological network, an assumption successfully applied in several works \cite{von2014boolean,fumia2013boolean,grieco2013integrative,moreno2013modeling,creixell2012navigating,rodriguez2012boolean,singh2012boolean,baverstock2011comparison,naldi2010diversity,thakar2010boolean,ge2009boolean,sahin2009modeling,davidich2008boolean,kervizic2008dynamical,faure2006dynamical,mendoza2006network,huang2000shape}. Assuming that a phenotype is an observable state, and thus relatively stable, of a biological system and assuming that the state of a biological system results from its dynamics, a phenotype is likely to correspond to an attractor. This assumption can be stated for any dynamical model but, in this work, only Boolean networks are considered. Reasons are that, in their most basic form, Boolean networks do not require quantitative information \cite{wynn2012logic} and that quantitative information is often not easy to obtain due to experimental limitations, particularly at the subcellular scale, the scale where drugs interact with their targets. Moreover, since synchronous Boolean networks are easier to compute than asynchronous ones \cite{garg2008synchronous}, this work only considers synchronous Boolean networks. This does not exclude the possibility, at a later stage, to extend the algorithm for both synchronous and asynchronous updating schemes.

For a biological network involved in a disease, two possible variants are considered: the physiological variant, exhibited by healthy organisms, which produces physiological phenotypes, and the pathological variant, exhibited by ill organisms, which produces pathological phenotypes or which fails to produce physiological ones. A physiological phenotype does not impair life quantity\slash quality whereas a pathological phenotype does. It should be noted that the loss of a physiological phenotype is also a pathological condition. The physiological and pathological variants differ in that the latter results from the occurrence of some alterations known to be responsible for disorders. With a pathological variant, there are two non-exclusive pathological scenarios: pathological phenotypes are gained or physiological phenotypes are lost.

The primary goal of the proposed algorithm is to identify, in a pathological variant, target combinations together with the perturbations to apply on them, here called bullets, which render it unable to exhibit pathological phenotypes. The secondary goal is to classify the obtained bullets according to their ability at rendering the pathological variant able to exhibit the previously lost physiological phenotypes, if any.

It should be noted that this work fits into the encompassing field investigating how to control biological systems, a field with tremendous applications in biomedicine. Several endeavors based on qualitative modeling approaches have been made in this way \cite{campbell2014stabilization,qiu2014control,srihari2014evolution,chen2013finding,kobayashi2013optimal,kobayashi2012symbolic}, demonstrating its utility in investigating how to take control over pathologically disturbed biological systems.

\section{Methods}
This section introduces some basic principles, namely biological and Boolean networks, defines some concepts and then describes the proposed algorithm. An example network to illustrate how it works plus a case study to illustrate its intended applications are also described. Finally, details about implementation and code availability are mentioned.

\subsection{Basic principles}
\subsubsection{Biological networks}
A biological network is a way to conceptualize a set of interacting biological entities where entities are represented by nodes and interactions by edges \cite{zhu2007getting,barabasi2004network}. It is based on graph theory \cite{bronshtein2007graphtheory,huber2007graphs,mason2007graph}, thus bringing formal tools to encode information about biological systems, particularly their topology \cite{ma2009insights}. Moreover, being graphs, biological networks offer a convenient visualization \cite{larkin1987diagram} of the complex interconnections lying in biological systems. As said Napoleon Bonaparte:
\begin{quote}
``A good sketch is better than a long speech.''
\end{quote}

Mathematically, a network can be seen as a digraph $G=(V,E)$ where $V=\lbrace v_{1},\dotsc,v_{n}\rbrace$ is the set of cardinality $n$ containing exactly all the nodes $v_{i}$ of the network and where $E=\lbrace (v_{i,1},v_{j,1}),\dotsc,(v_{i,m},v_{j,m})\rbrace \subseteq V^{2}$ is the set of cardinality $m$ containing exactly all the edges $(v_{i},v_{j})$ of the network. In practice, nodes represent entities and edges represent binary relations $R\subseteq V^{2}$ involving them: $v_{i}\ R\ v_{j}$. For example, in gene regulatory networks, nodes represent gene products and edges represent gene expression modulations \cite{xiao2009tutorial,rockett2006gene}.

\subsubsection{Boolean networks}
While being conceptually simple, Boolean networks \cite{saadatpour2012boolean} are able to predict and reproduce features of biological systems and then to bring relevant insights \cite{albert2014boolean,wang2012boolean,albert2008boolean,bornholdt2008boolean,huang2001genomics,leclerco1983boolean}. This makes them an attractive and efficient approach, especially when the complexity of biological systems renders quantitative approaches unfeasible due to the amount of quantitative details they require. As their name indicates, Boolean networks are based on Boolean logic \cite{boole1847mathematical} and, like biological networks, are also based on graph theory: nodes represent Boolean variables and edges represent interdependencies between them.

Mathematically, a Boolean network is a network where nodes are Boolean variables $x_{i}$ and where edges $(x_{i},x_{j})$ represent the binary $is\ input\ of$ relation: $x_{i}\ is\ input\ of\ x_{j}$. Each $x_{i}$ has $b_{i}\in [\![0,n]\!]$ inputs $x_{i,1},\dotsc,x_{i,b_{i}}$. The variables which are not inputs of $x_{i}$ have no direct influence on it. If $b_{i}=0$ then $x_{i}$ is a parameter and does not depend on other variables. At each iteration $k\in [\![k_{0},k_{end}]\!]$ of the simulation, the value $x_{i}(k)\in \lbrace 0,1\rbrace$ of each $x_{i}$ is updated to the value $x_{i}(k+1)$ using a Boolean function $f_{i}$ and the values $x_{i,1}(k),\dotsc,x_{i,b_{i}}(k)$ of its inputs, as in the following pseudocode:

\begin{algorithmic}[1]
\FOR{$k\in [\![k_{0},k_{end}-1]\!]$}
    \STATE $x_{1}(k+1)=f_{1}(x_{1,1}(k),\dotsc,x_{1,b_{1}}(k))$
    \STATE \dots
    \STATE $x_{n}(k+1)=f_{n}(x_{n,1}(k),\dotsc,x_{n,b_{n}}(k))$
\ENDFOR
\end{algorithmic}
which can be written in a more concise form:

\begin{algorithmic}[1]
\FOR{$k\in [\![k_{0},k_{end}-1]\!]$}
    \STATE $\boldsymbol{x}(k+1)=\boldsymbol{f}(\boldsymbol{x}(k))$
\ENDFOR
\end{algorithmic}
where $\boldsymbol{f}=(f_{1},\dotsc,f_{n})$ is the Boolean transition function and $\boldsymbol{x}=(x_{1},\dotsc,x_{n})$ is the state vector. The value $\boldsymbol{x}(k)=(x_{1}(k),\dotsc,x_{n}(k))\in \lbrace 0,1\rbrace^{n}$ of $\boldsymbol{x}$ at $k$ belongs to the state space $S=\lbrace 0,1\rbrace^{n}$ which is the set of cardinality $2^{n}$ containing exactly all the possible states.

If the values of all the $x_{i}$ are updated simultaneously at each $k$ then the network is synchronous, otherwise it is asynchronous. With synchronous Boolean networks, $\boldsymbol{x}(k)$ has a unique possible successor $\boldsymbol{x}(k+1)$: synchronous Boolean networks are deterministic. In the particular case where $k=k_{0}$, $\boldsymbol{x}(k_{0})=\boldsymbol{x}_{0}$ is the initial state and, in deterministic dynamical systems, determines entirely the trajectory $w=(\boldsymbol{x}(k_{0}),\dotsc,\boldsymbol{x}(k_{end}))$. In this work, it is assumed that $k_{0}=1$, so $w$ is a sequence of length $k_{end}$ resulting from the iterative computation of $\boldsymbol{x}(k)$ from $k_{0}$ up to $k_{end}$. This iterative computation can be seen as the discretization of a time interval: Boolean networks are discrete dynamical systems as they simulate discretely the time course of the state vector.

The set $A=\lbrace a_{1},\dotsc,a_{p}\rbrace$ of cardinality $p$ containing exactly all the attractors $a_{i}$ is called the attractor set. Due to the determinism of synchronous Boolean networks, all the attractors are cycles. A cycle is a sequence $(\boldsymbol{x}_{1},\dotsc,\boldsymbol{x}_{q})$ of length $q$ such that $\forall j\in [\![1,q]\!],\ \boldsymbol{x}_{j+1}=\boldsymbol{f}(\boldsymbol{x}_{j})$ and $\boldsymbol{x}_{q+1}=\boldsymbol{x}_{1}$: once the system reaches a state $\boldsymbol{x}_{j}$ belonging to a cycle, it successively visits its states $\boldsymbol{x}_{j+1},\dotsc,\boldsymbol{x}_{q},\boldsymbol{x}_{1},\dotsc,\boldsymbol{x}_{j}$ for infinity. In the particular case where $q=1$, $a_{i}$ is a point attractor. The set $B_{i}\subseteq S$ containing exactly all the $\boldsymbol{x}\in S$ from which $a_{i}$ can be reached is called its basin of attraction. With deterministic dynamical systems, the family of sets $(B_{1},\dotsc,B_{p})$ constitutes a partition of $S$.

\subsection{Definitions}
Some concepts used in this work should be formally defined.
\begin{itemize}
\item \textbf{physiological phenotype}: A phenotype which does not impair the life quantity\slash quality of the organism which exhibits it.
\item \textbf{pathological phenotype}: A phenotype which impairs the life quantity\slash quality of the organism which exhibits it.
\item \textbf{variant (of a biological network)}: Given a biological network of interest, a variant is one of its versions, namely the network plus eventually some modifications. It should be noted that this does not exclude the possibility that a variant can be the network of interest as is.
\item \textbf{physiological variant}: A variant which produces only physiological phenotypes. It is the biological network of interest as it should be, namely the one of healthy organisms.
\item \textbf{pathological variant}: A variant which produces at least one pathological phenotype or which fails to produce at least one physiological phenotype. It is a dysfunctional version of the biological network of interest, namely a version found in ill organisms.
\item \textbf{physiological attractor set}: The attractor set $A_{physio}$ of the physiological variant.
\item \textbf{pathological attractor set}: The attractor set $A_{patho}$ of the pathological variant.
\item \textbf{physiological Boolean transition function}: The Boolean transition function $\boldsymbol{f}_{physio}$ of the physiological variant.
\item \textbf{pathological Boolean transition function}: The Boolean transition function $\boldsymbol{f}_{patho}$ of the pathological variant.
\item \textbf{run}: An iterative computation of $\boldsymbol{x}(k)$ starting from an $\boldsymbol{x}_{0}$ until an $a_{i}$ is reached. It returns $w=(\boldsymbol{x}(k_{0}),\dotsc,\boldsymbol{x}(k_{end}))$ where $k_{end}$ depends on when $a_{i}$ is reached, and then on $\boldsymbol{x}_{0}$.
\item \textbf{physiological attractor}: An $a_{i}$ such that $a_{i}\in A_{physio}$.
\item \textbf{pathological attractor}: An $a_{i}$ such that $a_{i}\notin A_{physio}$.
\item \textbf{modality}: The functional perturbation $moda_{i}$ applied on a node $v_{j}\in V$ of the network, either activating ($moda_{i}=1$) or inactivating ($moda_{i}=0$): at each $k$, $moda_{i}$ overwrites $f_{j}(\boldsymbol{x}(k))$ making $x_{j}(k+1)=moda_{i}$.
\item \textbf{target}: A node $targ_{i}\in V$ of the network on which a $moda_{i}$ is applied.
\item \textbf{bullet}: A couple $(c_{targ},c_{moda})$ where $c_{targ}=(targ_{1},\dotsc,targ_{r})$ is a combination without repetition of $targ_{i}$ and where $c_{moda}=(moda_{1},\dotsc,moda_{r})$ is an arrangement with repetition of $moda_{i}$, $r\in [\![1,n]\!]$ being the number of targets in the bullet. Here, $moda_{i}$ is intended to be applied on $targ_{i}$.
\item \textbf{therapeutic bullet}: A bullet which makes $A_{patho}\subseteq A_{physio}$.
\item \textbf{silver bullet}: A therapeutic bullet which makes $A_{patho}\varsubsetneq A_{physio}$.
\item \textbf{golden bullet}: A therapeutic bullet which makes $A_{patho}=A_{physio}$.
\end{itemize}
The assumed link between phenotypes and attractors is the reason why attractors are qualified as either physiological or pathological according to the phenotype they produce. This is also the reason why, in this work, target identification aims at manipulating attractor sets of pathological variants.

\subsection{Steps of the algorithm}
The algorithm has two goals: i) finding therapeutic bullets, and ii) classifying them as either golden or silver. A therapeutic bullet makes the pathological variant unable at reaching pathological attractors, that is $A_{patho}\subseteq A_{physio}$. If such a bullet is applied on a pathological variant, the organism bearing it no longer exhibits the associated pathological phenotypes. However, a therapeutic bullet does not necessarily preserve\slash restore the physiological attractors. If a therapeutic bullet preserves\slash restores the physiological attractors, that is if $A_{patho}=A_{physio}$, then it is a golden one, but if $A_{patho}\varsubsetneq A_{physio}$ then it is a silver one.

Given a physiological and a pathological variant, that is $\boldsymbol{f}_{physio}$ and $\boldsymbol{f}_{patho}$, the algorithm follows five steps:
\begin{enumerate}
\item with $\boldsymbol{f}_{physio}$ it computes the control attractor set $A_{physio}$
\item it generates bullets and, for each of them, it performs the three following steps
\item with $\boldsymbol{f}_{patho}$ plus the bullet, it computes the variant attractor set $A_{patho}$
\item it assesses the therapeutic potential of the bullet by comparing $A_{physio}$ and $A_{patho}$ to detect pathological attractors
\item if the bullet is therapeutic then it is classified as either golden or silver by comparing $A_{physio}$ and $A_{patho}$ for equality
\end{enumerate}
These steps can be written in pseudocode as:

\begin{algorithmic}[1]
\STATE with $\boldsymbol{f}_{physio}$ compute $A_{physio}$
\STATE generate $bullet\_set$
\FOR{$bullet\in bullet\_set$}
    \STATE with $\boldsymbol{f}_{patho}$ plus $bullet$ compute $A_{patho}$
    \IF{$A_{patho}\subseteq A_{physio}$}
        \STATE $bullet$ is therapeutic
        \IF {$A_{patho}=A_{physio}$}
            \STATE $bullet$ is golden
        \ELSE
            \STATE $bullet$ is silver
        \ENDIF
    \ENDIF
\ENDFOR
\end{algorithmic}
The algorithm is described step by step but can be found as one block of pseudocode in \hyperref[appendix1]{\textit{Appendix \ref*{appendix1}}} page \pageref{appendix1}.

\subsubsection{Step 1: computing $A_{physio}$}
First of all, $A_{physio}$ has to be computed since it is the control and, as such, determines what is pathological. To do so, runs are performed with $\boldsymbol{f}_{physio}$ and the reached $a_{i}$ are stored in $A_{physio}$. However, $\boldsymbol{x}_{0}\in S$ and $card\ S$ increases exponentially with $n$. Even for reasonable values of $n$, $card\ S$ explodes: more than $1\ 000\ 000$ possible $\boldsymbol{x}_{0}$ for $n=20$. One solution ensuring that all the $a_{i}$ are reached is to start a run from each of the possible $\boldsymbol{x}_{0}$, that is from each of the $\boldsymbol{x}\in S$. Practically, this is unfeasible for an arbitrary value of $n$ since the required computational capacity can be too demanding. For example, assuming that a run requires $1$ millisecond and that $n=50$, performing a run from each of the $2^{50}$ $\boldsymbol{x}\in S$ requires nearly $36\ 000$ years.

Given that with deterministic dynamical systems $(B_{1},\dotsc,B_{p})$ is a partition of $S$, a solution is to select a subset $D\subseteq S$ of a reasonable cardinality containing the $\boldsymbol{x}_{0}$ to start from. In this work, $D$ is randomly selected from a uniform distribution. The stumbling block of this solution is that it does not ensure that at least one $\boldsymbol{x}_{0}$ per $B_{i}$ is selected and then does not ensure that all the $a_{i}$ are reached. This stumbling block holds only if $card\ D<card\ S$.

Again given that synchronous Boolean networks are deterministic, if a run visits a state already visited during a previous run then its destination, that is the reached attractor, is already found. If so, the run can be stopped and the algorithm can jump to the next one. To implement this, the previous trajectories are stored in a set $H$, the history, and at each $k$ the algorithm checks if $\exists w\in H\colon \boldsymbol{x}(k)\in w$. If this check is positive then the algorithm jumps to the next run.

To detect the attractors, since with deterministic dynamical systems they are cycles, the algorithm checks at each $k$ if $\boldsymbol{x}(k+1)$ is an already visited state of the current run, namely if $\exists k'\in [\![1,k]\!]\colon \boldsymbol{x}(k+1)=\boldsymbol{x}(k')$. If this check is positive then $a_{i}=(\boldsymbol{x}(k'),\dotsc,\boldsymbol{x}(k))$.

This step can be written in pseudocode as:

\begin{algorithmic}[1]
\STATE \textbf{prompt} $card\ D$
\STATE $card\ D=min(card\ D,2^{n})$
\STATE generate $D\subseteq S$
\STATE $H=\lbrace \rbrace$
\STATE $A_{physio}=\lbrace \rbrace$
\FOR{$x_{0}\in D$}
    \STATE $k=1$
    \STATE $\boldsymbol{x}(k)=x_{0}$
    \WHILE{\TRUE}
        \IF{$\exists w\in H\colon \boldsymbol{x}(k)\in w$}
            \STATE \textbf{break}
        \ENDIF
        \STATE $\boldsymbol{x}(k+1)=\boldsymbol{f}_{physio}(\boldsymbol{x}(k))$
        \IF{$\exists k'\in [\![1,k]\!]\colon \boldsymbol{x}(k+1)=\boldsymbol{x}(k')$}
            \STATE $A_{physio}=A_{physio}\cup \lbrace (\boldsymbol{x}(k'),\dotsc,\boldsymbol{x}(k))\rbrace$
            \STATE \textbf{break}
        \ENDIF
        \STATE $k=k+1$
    \ENDWHILE
    \STATE $H=H\cup \lbrace (\boldsymbol{x}(1),\dotsc,\boldsymbol{x}(k))\rbrace$
\ENDFOR
\RETURN $A_{physio}$
\STATE do step 2
\end{algorithmic}
Line 2 catches the mistake $card\ D>card\ S$.

It should be noted that the purpose of this work is not to propose an algorithm for finding Boolean network attractors since advanced algorithms for such tasks are already published \cite{guo2014parallel,berntenis2013detection,zheng2013efficient,dubrova2011sat,ay2009scalable}. The purpose is to propose a computational methodology exploiting Boolean network attractors for \textit{in silico} target identification, a methodology which requires \textit{de facto} these attractors to be found. This point is discussed in the \hyperref[conclusion]{\textit{Conclusion}} section page \pageref{conclusion}.

\subsubsection{Step 2: generating bullets}
Bullets are candidate perturbations to apply on the pathological variant to make it unable at reaching pathological attractors and then unable at producing pathological phenotypes. Generating a bullet requires a choice of $targ_{i}\in V$ and associated $moda_{i}\in \lbrace 0,1\rbrace$. In this work, there is no sequencing in target engagement nor in modality application. This means that, given a bullet and during a given run, all the $moda_{i}$ are applied on their corresponding $targ_{i}$ throughout the run. As a consequence, for a given bullet, choosing the same $targ_{i}$ more than once is senseless while it is possible to choose the same $moda_{i}$ for more than one $targ_{i}$. Therefore, a bullet is a combination $c_{targ}$ without repetition of $targ_{i}$ together with an arrangement $c_{moda}$ with repetition of $moda_{i}$.

If bullets containing $r$ targets have to be generated then there are $n!/(r!\cdot (n-r)!)$ possible $c_{targ}$ and, for each of them, there are $2^{r}$ possible $c_{moda}$. This raises the same computational difficulty than with the state space explosion since there are $(n!\cdot 2^{r})/(r!\cdot (n-r)!)$ possible bullets. For example, with $n=50$ and $r=3$ there are more than $150\ 000$ possible bullets. Knowing that the algorithm, as explained below, computes one attractor set per bullet, the computation time becomes practically unfeasible.

To overcome this barrier, the algorithm asks for $r$ as an interval $[\![r_{min},r_{max}]\!]$, asks for a maximum number $max_{targ}$ of $c_{targ}$ to generate and asks for a maximum number $max_{moda}$ of $c_{moda}$ to test for each $c_{targ}$. The algorithm then generates a set $C_{targ}$ of $c_{targ}$ with $card\ C_{targ}\leq max_{targ}$ by randomly selecting, from a uniform distribution and without repetition, nodes in the network. In the same way, the algorithm generates a set $C_{moda}$ of $c_{moda}$ with $card\ C_{moda}\leq max_{moda}$ by randomly choosing, from a uniform distribution and with repetition, modalities as either activating ($1$) or inactivating ($0$). The result is the bullets: per $r\in[\![r_{min},r_{max}]\!]$, a $C_{targ}$ together with a $C_{moda}$. As with the state space explosion, the stumbling block of this method is that it does not ensure that all the possible $c_{targ}$ together with all the possible $c_{moda}$ are tested. This stumbling block holds only if $max_{targ}<n!/(r!\cdot (n-r)!)$ or $max_{moda}<2^{r}$.

This step can be written in pseudocode as:

\begin{algorithmic}[1]
\STATE \textbf{prompt} $r_{min},r_{max},max_{targ},max_{moda}$
\STATE $r_{max}=min(r_{max},n)$
\STATE $golden\_set=\lbrace \rbrace$
\STATE $silver\_set=\lbrace \rbrace$
\FOR{$r\in [\![r_{min},r_{max}]\!]$}
    \STATE $max_{targ}^{r}=min(max_{targ},n!/(r!\cdot (n-r)!))$
    \STATE $max_{moda}^{r}=min(max_{moda},2^{r})$
    \STATE $C_{targ}=\lbrace \rbrace$
    \STATE $C_{moda}=\lbrace \rbrace$
    \WHILE{$card\ C_{targ}<max_{targ}^{r}$}
        \STATE generate $c_{targ}\notin C_{targ}$
        \STATE $C_{targ}=C_{targ}\cup \lbrace c_{targ}\rbrace$
    \ENDWHILE
    \WHILE{$card\ C_{moda}<max_{moda}^{r}$}
        \STATE generate $c_{moda}\notin C_{moda}$
        \STATE $C_{moda}=C_{moda}\cup \lbrace c_{moda}\rbrace$
    \ENDWHILE
    \STATE do steps 3 to 5
\ENDFOR
\RETURN $golden\_set,silver\_set$
\end{algorithmic}
Line 2 catches the mistake $r>n$. Lines 3 and 4 create the sets in which the therapeutic bullets found in step 4 are classified as either golden or silver in step 5. Lines 6 and 7 catch the mistake where $max_{targ}$ or $max_{moda}$ is greater than its maximum, which depends on $r$, hence the creation of $max_{targ}^{r}$ and $max_{moda}^{r}$ to preserve the initially supplied value. Lines 11 and 15 ensure that only new $c_{targ}$ and $c_{moda}$ are generated.

\subsubsection{Step 3: computing $A_{patho}$}
Having the control attractor set $A_{physio}$ and a bullet $(c_{targ},c_{moda})\in C_{targ}\times C_{moda}$, the algorithm computes the variant attractor set $A_{patho}$ under the effect of $(c_{targ},c_{moda})$ by almost the same way $A_{physio}$ is computed in step 1. However, $\boldsymbol{f}_{patho}$ is used instead of $\boldsymbol{f}_{physio}$ and $(c_{targ},c_{moda})$ is applied: at each $k$, $f_{j}(\boldsymbol{x}(k))$ is overwritten by $moda_{i}\in c_{moda}$, that is $x_{j}(k+1)=moda_{i}$, provided that $v_{j}=targ_{i}\in c_{targ}$. In order to apply all the generated bullets, the algorithm uses two nested $for$ loops. For each $c_{targ}\in C_{targ}$, it uses successively all the $c_{moda}\in C_{moda}$. For each $(c_{targ},c_{moda})$, the algorithm computes the corresponding $A_{patho}$ and does steps 4 and 5.

This step can be written in pseudocode as:

\begin{algorithmic}[1]
\FOR{$c_{targ}\in C_{targ}$}
    \FOR{$c_{moda}\in C_{moda}$}
        \STATE $H=\lbrace \rbrace$
        \STATE $A_{patho}=\lbrace \rbrace$
        \FOR{$x_{0}\in D$}
            \STATE $k=1$
            \STATE $\boldsymbol{x}(k)=x_{0}$
            \WHILE{\TRUE}
                \IF{$\exists w\in H\colon \boldsymbol{x}(k)\in w$}
                    \STATE \textbf{break}
                \ENDIF
                \STATE $\boldsymbol{x}(k+1)=\boldsymbol{f}_{patho}(\boldsymbol{x}(k))$
                \FOR{$targ_{i}\in c_{targ}$}
                    \FOR{$v_{j}\in V$}
                        \IF{$v_{j}=targ_{i}$}
                            \STATE $x_{j}(k+1)=moda_{i}$
                        \ENDIF
                    \ENDFOR
                \ENDFOR
                \IF{$\exists k'\in [\![1,k]\!]\colon \boldsymbol{x}(k+1)=\boldsymbol{x}(k')$}
                    \STATE $A_{patho}=A_{patho}\cup \lbrace (\boldsymbol{x}(k'),\dotsc,\boldsymbol{x}(k))\rbrace$
                    \STATE \textbf{break}
                \ENDIF
                \STATE $k=k+1$
            \ENDWHILE
            \STATE $H=H\cup \lbrace (\boldsymbol{x}(1),\dotsc,\boldsymbol{x}(k))\rbrace$
        \ENDFOR
        \STATE do step 4 and 5
    \ENDFOR
\ENDFOR
\end{algorithmic}
Lines 13--19 are where bullets are applied.

\subsubsection{Step 4: identifying therapeutic bullets}
To identify therapeutic bullets among the generated ones, for each $(c_{targ},c_{moda})$ tested in step 3 and once the corresponding $A_{patho}$ is obtained, the algorithm compares it with $A_{physio}$ to check if $A_{patho}\subseteq A_{physio}$. This check ensures that the pathological attractors are removed and that if new attractors appear then they are physiological. If this check is positive then the bullet is therapeutic and the algorithm pursues with step 5.

This step can be written in pseudocode as:

\begin{algorithmic}[1]
\IF{$A_{patho}\subseteq A_{physio}$}
    \STATE do step 5
\ENDIF
\end{algorithmic}

\subsubsection{Step 5: assessing therapeutic bullets}
Therapeutic bullets are qualified as either golden or silver according to their ability at making the pathological variant reaching the physiological attractors. All therapeutic bullets, being golden or silver, remove the pathological attractors without creating new ones, that is $A_{patho}\subseteq A_{physio}$. However, this does not imply that they preserve\slash restore the physiological attractors. A golden bullet preserves\slash restores the physiological attractors: $A_{patho}=A_{physio}$ whereas a silver bullet does not: $A_{patho}\varsubsetneq A_{physio}$.

In this setting, golden bullets are perfect therapies whereas silver bullets are not. However, since precious things are rare and just as gold is rarer than silver, finding golden bullets is less likely than finding silver ones. Indeed, given that more constraints are required for a therapeutic bullet to be golden, it is more likely that the found therapeutic bullets are silver, except in one case: $card\ A_{physio}=1$.

\begin{theorem}
\label{theorem}
If $card\ A_{physio}=1$ then therapeutic bullets are golden.
\end{theorem}

\begin{proof}
\begin{IEEEeqnarray*}{l}
(\text{therapeutic bullet}) \Rightarrow (A_{patho} \subseteq A_{physio})\hfill (1)\\
(1) \Rightarrow (A_{patho} \in \mathcal{P}(A_{physio}))\hfill (2)\\
(card\ A_{physio}=1) \Rightarrow (A_{physio}=\lbrace a \rbrace)\hfill (3)\\
(3) \Rightarrow (\mathcal{P}(A_{physio})=\lbrace \emptyset,\lbrace a \rbrace \rbrace)\hfill (4)\\
((2) \land (4)) \Rightarrow ((A_{patho}=\lbrace a \rbrace) \lor (A_{patho}=\emptyset))\hfill (5)\\
(\text{deterministic dynamical systems}) \Rightarrow (A \neq \emptyset)\hfill (6)\\
(6) \Rightarrow (A_{patho} \neq \emptyset)\hfill (7)\\
((5) \land (7)) \Rightarrow (A_{patho}=\lbrace a \rbrace)\hfill (8)\\
((3) \land (8)) \Rightarrow (A_{patho}=A_{physio})\hfill (9)\\
(9) \Rightarrow (\text{therapeutic bullet is golden})\hfill (10)
\end{IEEEeqnarray*}
\end{proof}

Practically, an organism bearing a pathological variant treated with a therapeutic bullet no longer exhibits the associated pathological phenotypes. Moreover, if the therapeutic bullet is golden then the organism exhibits the same phenotypes than its healthy counterpart. However, if the therapeutic bullet is silver then the organism fails to exhibit at least one physiological phenotype. With a silver bullet this is a matter of choice: what is the less detrimental between a silver bullet and no therapeutic bullet at all?

This step can be written in pseudocode as:

\begin{algorithmic}[1]
\IF{$A_{patho}=A_{physio}$}
    \STATE $golden\_set=golden\_set\cup \lbrace (c_{targ},c_{moda})\rbrace$
\ELSE
    \STATE $silver\_set=silver\_set\cup \lbrace (c_{targ},c_{moda})\rbrace$
\ENDIF
\end{algorithmic}

\subsection{Example network}
To illustrate the algorithm, it is used on a Boolean model of the mammalian cell cycle published by Faure \textit{et al} \cite{faure2006dynamical}. This model is chosen for several reasons: i) a synchronous updating is performed: to date, the algorithm focuses on synchronous Boolean networks, ii) a mammalian biological system is modeled: the closer to human physiology the model is the better it illustrates the intended applications, iii) the cell cycle is a at the heart of cancer: this gives relevancy to the example network, iv) the network comprises ten nodes: easily computable in face of its state space, and v) attractors are already computed: useful to validate the algorithm in finding them. A graphical representation of the example network is shown in \hyperref[cell_cycle]{\textit{Figure \ref*{cell_cycle}}} page \pageref{cell_cycle}. Below are the corresponding Boolean functions where, for the sake of readability, $x_{i}$ stands for $x_{i}(k)$ and $x_{i+}$ stands for $x_{i}(k+1)$:

\begin{footnotesize}
\begin{IEEEeqnarray*}{r C l}
CycD_{+}&=&CycD\\
Rb_{+}&=&(\lnot CycD \land \lnot CycE \land \lnot CycA \land \lnot CycB) \lor (p27 \land \lnot CycD \land \lnot CycB)\\
E2F_{+}&=&(\lnot Rb \land \lnot CycA \land \lnot CycB) \lor (p27 \land \lnot Rb \land \lnot CycB)\\
CycE_{+}&=&E2F \land \lnot Rb\\
CycA_{+}&=&(E2F \land \lnot Rb \land \lnot Cdc20 \land \lnot(Cdh1 \land UbcH10))\\
&&\lor (CycA \land \lnot Rb \land \lnot Cdc20 \land \lnot(Cdh1 \land UbcH10))\\
p27_{+}&=&(\lnot CycD \land \lnot CycE \land \lnot CycA \land \lnot CycB)\\
&&\lor (p27 \land \lnot(CycE \land CycA) \land \lnot CycB \land \lnot CycD)\\
Cdc20_{+}&=&CycB\\
Cdh1_{+}&=&(\lnot CycA \land \lnot CycB) \lor Cdc20 \lor (p27 \land \lnot CycB)\\
UbcH10_{+}&=&\lnot Cdh1 \lor (Cdh1 \land UbcH10 \land (Cdc20 \lor CycA \lor CycB))\\
CycB_{+}&=&\lnot Cdc20 \land \lnot Cdh1
\end{IEEEeqnarray*}
\end{footnotesize}

\begin{figure*}[h!]
\begin{center}
\includegraphics[width=0.75\textwidth]{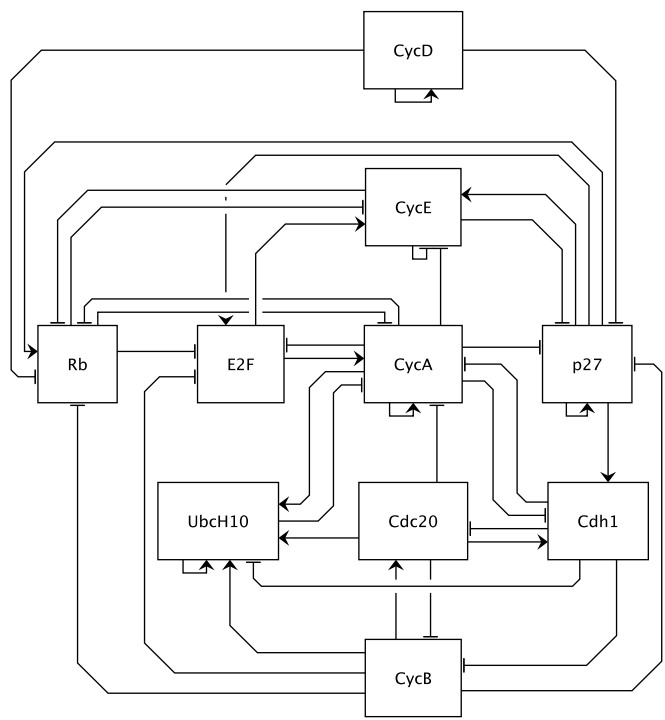}
\caption[Example network]{\label{cell_cycle} Graphical representation of the example network adapted from \cite{faure2006dynamical}. CDKs (cyclin-dependent kinases) are the catalytic partners of cyclins and, in this model, are not explicitly shown since the activity of CDK-cyclin complexes essentially depends on cyclins. Furthermore, the inhibition of E2F by Rb is modeled by opposing Rb to the effects of E2F on its targets. The same applies to the inhibition of CycE and CycA by p27. For a complete description of the model, see \cite{faure2006dynamical}. CycD: CDK4\slash 6-cyclin D complex, input of the model, initiates the cell cycle, activated by positive signals such as growth factors; CycE: CDK2-cyclin E complex; CycA: CDK2-cyclin A complex; CycB: CDK1-cyclin B complex; Rb: retinoblastoma protein, a tumor suppressor; E2F: a family of transcription factors divided into activator and repressor members, in this model E2F represents the activator members; p27: p27\slash Kip1, a CKI (CDK inhibitor); Cdc20: an APC (Anaphase Promoting Complex, an E3 ubiquitin ligase) activator; Cdh1: an APC activator; UbcH10: an E2 ubiquitin conjugating enzyme.}
\end{center}
\end{figure*}

Having the example network, two variants are needed: the physiological one and the pathological one. The physiological variant is the network as is while the pathological variant is the network plus a constitutive activation\slash inactivation of at least one of its nodes. For simplicity, and given the relatively small number of entities, only one is chosen: the retinoblastoma protein Rb for which a constitutive inactivation is applied. To implement this, the corresponding $f_{i}$ becomes:

\begin{IEEEeqnarray*}{c}
Rb(k+1)=0
\end{IEEEeqnarray*}
in $\boldsymbol{f}_{patho}$. Rb is chosen because its inactivation occurs in many cancers \cite{sherr2002rb}. Therefore, a network bearing a constitutive inactivation of it should be a relevant example of a pathological variant.

\subsection{Case study}
To illustrate the intended usage of the proposed methodology, the algorithm is used on a Boolean model of the Fanconi Anemia\slash Breast Cancer (FA\slash BRCA) pathway published by Rodriguez \textit{et al} \cite{rodriguez2012boolean}. This model is chosen for several reasons: i) two pathological conditions are studied: required for a case study of an \textit{in silico} target identification, ii) the physiological and pathological variants are clearly described: required by the algorithm, iii) it is nearly three times bigger than the example network: representative of a more comprehensive biological model while remaining computationally tractable, iv) synchronous updating is used: to date, the algorithm focuses on synchronous Boolean networks, and v) attractors are already interpreted in terms of phenotypes.

The FA\slash BRCA pathway is dedicated to DNA repair, more precisely to interstrand cross-link (ICL) removal. As expected with any DNA repair impairment, individuals suffering from FA\slash BRCA pathway malfunction are subjected to increased risk of cancer, such as in Fanconi anemia, a rare genetic disorder causing bone marrow failure, congenital abnormalities and increased risk of cancer \cite{schwartz2010susceptibility,auerbach2009fanconi,de2009genetic}. Rodriguez \textit{et al} propose a Boolean model comprising the FA\slash BRCA pathway and three types of DNA damages commonly observed in Fanconi anemia, namely ICLs, double-strand breaks (DSBs) and DNA adducts (ADDs). It should be noted that the ICL repair process creates DSBs and ADDs before removing them, thus leaving an undamaged DNA ready for the cell cycle. For a complete description of the model, see \cite{rodriguez2012boolean}. Below are the corresponding Boolean functions where, for the sake of readability, $x_{i}$ stands for $x_{i}(k)$ and $x_{i+}$ stands for $x_{i}(k+1)$:

\begin{footnotesize}
\begin{IEEEeqnarray*}{r C l}
ICL_{+}&=&ICL \land \lnot DSB\\
FANCM_{+}&=&ICL \land \lnot CHKREC\\
FAcore_{+}&=&FANCM \land (ATR \lor ATM) \land \lnot CHKREC\\
FANCD2I_{+}&=&FAcore \land ((ATM \lor ATR) \lor (H2AX \land DSB)) \land \lnot USP1\\
MUS81_{+}&=&ICL\\
FANCJBRCA1_{+}&=&(ICL \lor ssDNARPA) \land (ATM \lor ATR)\\
XPF_{+}&=&(MUS81 \land p53 \land \lnot (FAcore \land FANCD2I \land FAN1))\\
&&\lor (MUS81 \land \lnot FANCM)\\
FAN1_{+}&=&MUS81 \land FANCD2I\\
ADD_{+}&=&(ADD \lor (MUS81 \land (FAN1 \lor XPF))) \land \lnot PCNATLS\\
DSB_{+}&=&(DSB \lor FAN1 \lor XPF) \land \lnot (NHEJ \lor HRR)\\
PCNATLS_{+}&=&(ADD \lor (ADD \land FAcore)) \land \lnot (USP1 \lor FAN1)\\
MRN_{+}&=&DSB \land ATM \land \lnot ((KU \land FANCD2I) \lor RAD51 \lor CHKREC)\\
BRCA1_{+}&=&DSB \land (ATM \lor CHK2 \lor ATR) \land \lnot CHKREC\\
ssDNARPA_{+}&=&DSB \land ((FANCD2I \land FANCJBRCA1) \lor MRN) \land \lnot (RAD51 \lor KU)\\
FANCD1N_{+}&=&(FANCD2I \land ssDNARPA \land \lnot CHKREC)\\
&&\lor (ssDNARPA \land BRCA1)\\
RAD51_{+}&=&ssDNARPA \land FANCD1N \land \lnot CHKREC\\
HRR_{+}&=&DSB \land RAD51 \land FANCD1N \land BRCA1 \land \lnot CHKREC\\
USP1_{+}&=&((FANCD1N \land FANCD2I) \lor PCNATLS) \land \lnot FANCM\\
KU_{+}&=&DSB \land \lnot (MRN \lor FANCD2I \lor CHKREC)\\
DNAPK_{+}&=&(DSB \land KU) \land \lnot CHKREC\\
NHEJ_{+}&=&((DSB \land DNAPK \land KU) \land \lnot (ATM \land ATR))\\
&&\lor (\lnot ((FANCJBRCA1 \land ssDNARPA) \lor CHKREC)\\
&&\land DSB \land DNAPK \land XPF)\\
ATR_{+}&=&(ssDNARPA \lor FANCM \lor ATM) \land \lnot CHKREC\\
ATM_{+}&=&(ATR \lor DSB) \land \lnot CHKREC\\
p53_{+}&=&(((ATM \land CHK2) \lor (ATR \land CHK1)) \lor DNAPK) \land \lnot CHKREC\\
CHK1_{+}&=&(ATM \lor ATR \lor DNAPK) \land \lnot CHKREC\\
CHK2_{+}&=&(ATM \lor ATR \lor DNAPK) \land \lnot CHKREC\\
H2AX_{+}&=&DSB \land (ATM \lor ATR \lor DNAPK) \land \lnot CHKREC\\
CHKREC_{+}&=&((PCNATLS \lor NHEJ \lor HRR) \land \lnot DSB)\\
&&\lor ((\lnot ADD) \land (\lnot ICL) \land (\lnot DSB) \land \lnot CHKREC)
\end{IEEEeqnarray*}
\end{footnotesize}

The physiological variant is the FA\slash BRCA pathway model as is. To it, Rodriguez \textit{et al} propose two pathological variants, here called \textit{patho1} and \textit{patho2}, modeling two mutations involving genes of the FA\slash BRCA pathway. These mutations are observed in patients suffering from Fanconi anemia \cite{neveling2009genotype}. The first one involves the FANCA gene, corresponding to the $FAcore$ variable, and the second one involves the FANCD1\slash BRCA2 or FANCN\slash PALB2 gene, corresponding to the $FANCD1N$ variable. These mutations are of loss-of-function kind: to simulate them, the corresponding $f_{i}$ become

\begin{IEEEeqnarray*}{c}
FAcore(k+1)=0
\end{IEEEeqnarray*}
for FANCA gene null mutation in $\boldsymbol{f}_{patho1}$ and

\begin{IEEEeqnarray*}{c}
FANCD1N(k+1)=0
\end{IEEEeqnarray*}
for FANCD1\slash BRCA2 or FANCN\slash PALB2 gene null mutation in $\boldsymbol{f}_{patho2}$.

\subsection{Implementation}
The algorithm is implemented in Fortran compiled with GFortran\footnote{\url{http://www.gnu.org/software/gcc/fortran/}}. The code is available on GitHub\footnote{\url{https://github.com/}} at \url{https://github.com/arnaudporet/kali-targ}.

\section{Results}
In this section, results produced with the algorithm on the example network are exposed to illustrate how it works. Next, results produced with the algorithm on the case study are exposed to illustrate its intended applications for target identification.

\subsection{Results of step 1}
Owing to the relatively small size of the example network, $card\ D$ is set to $card\ S=1024$. Since $card\ D=card\ S$, all the attractors are found. Attractors are presented as matrices where, for an attractor of length $q$, lines correspond to the $x_{i}(k),\ k\in [\![1,q]\!]$, and columns to $\boldsymbol{x}(k)$. The algorithm returns the following attractors:

\begin{center}
$a_{1}=$
\begin{tabular}{l|l l l l l l l}
$CycD$&$1$&$1$&$1$&$1$&$1$&$1$&$1$\\
$Rb$&$0$&$0$&$0$&$0$&$0$&$0$&$0$\\
$E2F$&$0$&$1$&$1$&$1$&$0$&$0$&$0$\\
$CycE$&$0$&$0$&$1$&$1$&$1$&$0$&$0$\\
$CycA$&$0$&$0$&$0$&$1$&$1$&$1$&$1$\\
$p27$&$0$&$0$&$0$&$0$&$0$&$0$&$0$\\
$Cdc20$&$1$&$0$&$0$&$0$&$0$&$0$&$1$\\
$Cdh1$&$1$&$1$&$1$&$1$&$0$&$0$&$0$\\
$UbcH10$&$1$&$1$&$0$&$0$&$0$&$1$&$1$\\
$CycB$&$0$&$0$&$0$&$0$&$0$&$1$&$1$
\end{tabular}
\end{center}

\begin{center}
$a_{2}=$
\begin{tabular}{l|l}
$CycD$&$0$\\
$Rb$&$1$\\
$E2F$&$0$\\
$CycE$&$0$\\
$CycA$&$0$\\
$p27$&$1$\\
$Cdc20$&$0$\\
$Cdh1$&$1$\\
$UbcH10$&$0$\\
$CycB$&$0$
\end{tabular}
\end{center}
each of them attracting $50\%$ of the $\boldsymbol{x}\in S$ under $\boldsymbol{f}_{physio}$. Then, $A_{physio}=\lbrace a_{1},a_{2}\rbrace$ and corresponds to the results obtained by Faure \textit{et al}. In terms of phenotypes, $a_{1}$ corresponds to cell cycle whereas $a_{2}$ corresponds to quiescence.

\subsection{Results of steps 2 to 5}
Results of steps 2 to 5 are grouped since only the therapeutic bullets found in step 4 and classified in step 5 are returned. The algorithm is launched with $r_{min}=1$ and $r_{max}=2$. Due to the relatively small size of the example network, $max_{targ}$ and $max_{moda}$ are set to their maximum, namely $max_{targ}=45$ and $max_{moda}=4$. Consequently, all the possible bullets made of $1$ to $2$ targets are tested. The algorithm returns the following therapeutic bullets:

\begin{center}
\begin{tabular}{l l|l}
$+CycD$&&silver\\
$+CycD$&$-p27$&silver\\
$-CycD$&$+Rb$&silver\\
$+CycD$&$-Rb$&silver
\end{tabular}
\end{center}
where $+$ means therapeutic activation and $-$ means therapeutic inactivation. It should be noted that no golden bullets are found, an unsurprising result since they are rarer than silver ones.

Given these results, therapeutic activation of Rb, which is inactivated in the pathological variant, is not enough to remove the pathological attractors. Indeed, as seen in the third bullet, therapeutic activation of Rb must be accompanied by therapeutic inactivation of CycD. To better illustrate what is performed to obtain these therapeutic bullets, below is $A_{patho}$ without any bullet:

\begin{center}
$a_{3}=$
\begin{tabular}{l|l l l l l l l l}
$CycD$&$0$&$0$&$0$&$0$&$0$&$0$&$0$&$0$\\
$Rb$&$0$&$0$&$0$&$0$&$0$&$0$&$0$&$0$\\
$E2F$&$1$&$1$&$1$&$1$&$0$&$0$&$0$&$0$\\
$CycE$&$0$&$1$&$1$&$1$&$1$&$0$&$0$&$0$\\
$CycA$&$0$&$0$&$1$&$1$&$1$&$1$&$1$&$0$\\
$p27$&$1$&$1$&$1$&$0$&$0$&$0$&$0$&$0$\\
$Cdc20$&$0$&$0$&$0$&$0$&$0$&$0$&$1$&$1$\\
$Cdh1$&$1$&$1$&$1$&$1$&$0$&$0$&$0$&$1$\\
$UbcH10$&$1$&$0$&$0$&$0$&$0$&$1$&$1$&$1$\\
$CycB$&$0$&$0$&$0$&$0$&$0$&$1$&$1$&$0$
\end{tabular}
\end{center}

\begin{center}
$a_{4}=$
\begin{tabular}{l|l l l l l l l}
$CycD$&$1$&$1$&$1$&$1$&$1$&$1$&$1$\\
$Rb$&$0$&$0$&$0$&$0$&$0$&$0$&$0$\\
$E2F$&$1$&$1$&$1$&$0$&$0$&$0$&$0$\\
$CycE$&$0$&$1$&$1$&$1$&$0$&$0$&$0$\\
$CycA$&$0$&$0$&$1$&$1$&$1$&$1$&$0$\\
$p27$&$0$&$0$&$0$&$0$&$0$&$0$&$0$\\
$Cdc20$&$0$&$0$&$0$&$0$&$0$&$1$&$1$\\
$Cdh1$&$1$&$1$&$1$&$0$&$0$&$0$&$1$\\
$UbcH10$&$1$&$0$&$0$&$0$&$1$&$1$&$1$\\
$CycB$&$0$&$0$&$0$&$0$&$1$&$1$&$0$
\end{tabular}
\end{center}
each of these two attractors attracting $50\%$ of the $\boldsymbol{x}\in S$ under $\boldsymbol{f}_{patho}$. It should be noted that $a_{4}=a_{1}\in A_{physio}$: $a_{4}$ is a physiological attractor which also belongs to $A_{patho}$. Indeed, it is possible that the pathological variant exhibits physiological attractors: $A_{patho}$ is not the set containing exactly all the pathological attractors, it is the attractor set of the pathological variant, so $A_{physio}\cap A_{patho}\neq \emptyset$ is possible. However, $a_{3}\notin A_{physio}$: it is a pathological attractor and is what a therapeutic bullet, being golden or silver, is intended to remove.

Again to better illustrate what is performed to obtain these therapeutic bullets, below is $A_{patho}$ under the third bullet:

\begin{center}
\begin{tabular}{l|l}
$CycD$&$0$\\
$Rb$&$1$\\
$E2F$&$0$\\
$CycE$&$0$\\
$CycA$&$0$\\
$p27$&$1$\\
$Cdc20$&$0$\\
$Cdh1$&$1$\\
$UbcH10$&$0$\\
$CycB$&$0$
\end{tabular}
\end{center}
which is $a_{2}$. As expected for a therapeutic bullet, the pathological attractor $a_{3}$ is removed. However, the physiological attractor $a_{1}$ is not restored: the third therapeutic bullet is silver. Consequently, with this therapeutic bullet no cell cycle occurs and the only reachable phenotype is quiescence. While disabling the cell cycle of cancer cells is beneficial, disabling the cell cycle of healthy cells is not. As mentioned above, with silver bullets this is a matter of choice.

\subsection{Results of the case study}
With the case study, $card\ S=268\ 435\ 456$: computing attractors from all the $\boldsymbol{x} \in S$ becomes too demanding. Indeed, it should be recalled that the algorithm computes one attractor set per bullet, namely $A_{patho}$ under the tested bullet. Consequently, $card\ D$ is set to a more reasonable value: $card\ D=10\ 000$. Despite that $card\ D<card\ S$, it seems sufficient for the algorithm to find all the attractors, just as Rodriguez \textit{et al} whose the computation covers the whole state space. Below are the computed attractors:

\begin{itemize}
\item $A_{physio}=\lbrace a_{1}\rbrace$
\item $A_{patho1}=\lbrace a_{1}\rbrace$
\item $A_{patho2}=\lbrace a_{1},a_{2}\rbrace$, $a_{1}$ and $a_{2}$ attracting respectively $29.5\%$ and $70.5\%$ of the $\boldsymbol{x}\in D$ under $\boldsymbol{f}_{patho2}$
\end{itemize}
where

\begin{center}
$a_{1}=$
\begin{tabular}{l|l l}
$ICL$&$0$&$0$\\
$FANCM$&$0$&$0$\\
$FAcore$&$0$&$0$\\
$FANCD2I$&$0$&$0$\\
$MUS81$&$0$&$0$\\
$FANCJBRCA1$&$0$&$0$\\
$XPF$&$0$&$0$\\
$FAN1$&$0$&$0$\\
$ADD$&$0$&$0$\\
$DSB$&$0$&$0$\\
$PCNATLS$&$0$&$0$\\
$MRN$&$0$&$0$\\
$BRCA1$&$0$&$0$\\
$ssDNARPA$&$0$&$0$\\
$FANCD1N$&$0$&$0$\\
$RAD51$&$0$&$0$\\
$HRR$&$0$&$0$\\
$USP1$&$0$&$0$\\
$KU$&$0$&$0$\\
$DNAPK$&$0$&$0$\\
$NHEJ$&$0$&$0$\\
$ATR$&$0$&$0$\\
$ATM$&$0$&$0$\\
$p53$&$0$&$0$\\
$CHK1$&$0$&$0$\\
$CHK2$&$0$&$0$\\
$H2AX$&$0$&$0$\\
$CHKREC$&$0$&$1$
\end{tabular}\hfill
$a_{2}=$
\begin{tabular}{l|l}
$ICL$&$0$\\
$FANCM$&$0$\\
$FAcore$&$0$\\
$FANCD2I$&$0$\\
$MUS81$&$0$\\
$FANCJBRCA1$&$1$\\
$XPF$&$0$\\
$FAN1$&$0$\\
$ADD$&$0$\\
$DSB$&$1$\\
$PCNATLS$&$0$\\
$MRN$&$1$\\
$BRCA1$&$1$\\
$ssDNARPA$&$1$\\
$FANCD1N$&$0$\\
$RAD51$&$0$\\
$HRR$&$0$\\
$USP1$&$0$\\
$KU$&$0$\\
$DNAPK$&$0$\\
$NHEJ$&$0$\\
$ATR$&$1$\\
$ATM$&$1$\\
$p53$&$1$\\
$CHK1$&$1$\\
$CHK2$&$1$\\
$H2AX$&$1$\\
$CHKREC$&$0$
\end{tabular}
\end{center}
and their biological interpretation:

\begin{itemize}
\item $a_{1}$: cell cycle progression
\item $a_{2}$: cell cycle arrest
\end{itemize}

In physiological conditions, in case of a damaged DNA, cells repair it before performing the cell cycle, or die if repair fails. Such checkpoints enable cells to ensure genomic integrity by preventing damaged DNA to be replicated and then propagated \cite{bartek2007dna,ishikawa2006dna}. Otherwise, genetic instability may appears, potentially leading to cancer \cite{nakanishi2006genetic}. The results show that the physiological variant is able to ensure genomic integrity since its unique attractor is $a_{1}$ where $ICL=DSB=ADD=0$: DNA damages are repaired, if any, and the cell cycle can safely occur. Interestingly, the same physiological phenotype is computed for \textit{patho1} where $A_{patho1}=A_{physio}$. This suggests that cells bearing FANCA gene null mutation are nonetheless able to repair DNA. With \textit{patho2}, a pathological attractor appears: $a_{2}$, where $DSB=1$. This suggests that cells bearing FANCD1\slash BRCA2 or FANCN\slash PALB2 gene null mutation are unable to repair DSBs, explaining why $a_{2}$ corresponds to cell cycle arrest: DNA remains damaged. It should be noted that $a_{1}\in A_{patho2}$, suggesting that from some $\boldsymbol{x}_{0}$, that is under some conditions, such cells could be able to repair DNA. However, $a_{1}$ attracts only $29.5\%$ of the $\boldsymbol{x}\in D$ under $\boldsymbol{f}_{patho2}$, indicating that the pathological phenotype associated with $a_{2}$ is the most likely.

Altogether, according to the computed attractors and their phenotypic interpretation, and limited to the scope studied by the model of Rodriguez \textit{et al}, FANCA gene null mutation may not induce pathological phenotypes. However, with FANCD1\slash BRCA2 or FANCN\slash PALB2 gene null mutation, two phenotypes are predicted: a physiological one and a pathological one, the latter being the most likely. Therefore, the algorithm has to operate on \textit{patho2} to find bullets able to remove the pathological attractor $a_{2}$. By comprehensively testing all the bullets made of $1$ to $3$ targets, the algorithm returns the following results:

\begin{center}
\begin{tabular}{c|c|c}
&number of all possible bullets&number of therapeutic bullets\\ \hline
$r=1$&$56$&$1$ ($1.786\%$)\\
$r=2$&$1\ 512$&$20$ ($1.323\%$)\\
$r=3$&$26\ 208$&$191$ ($0.729\%$)
\end{tabular}
\end{center}
all therapeutic bullets being golden since $card\ A_{physio}=1$, as demonstrated in the \hyperref[theorem]{\textit{Theorem \ref*{theorem}}} page \pageref{theorem}. A list of the computed therapeutic bullets can be found in \hyperref[appendix2]{\textit{Appendix \ref*{appendix2}}} page \pageref{appendix2}. Given that in $a_{1}$, what the pathological variant is forced to reach by means of therapeutic bullets, almost all variables are valued at $0$, it is unsurprising that all targets in the computed therapeutic bullets have to be inhibited, that is set to $0$.

Below is the frequency of each node in the found therapeutic bullets:

\begin{center}
\begin{tabular}{c|c}
node&frequency in the found therapeutic bullets\\ \hline
$ATM$&$87.736\%$\\
$ICL$&$22.170\%$\\
$BRCA1$&$18.396\%$\\
$DSB$&$11.792\%$\\
$MRN$&$10.377\%$\\
$FANCM$&$9.906\%$\\
$ADD$&$9.906\%$\\
$FANCJBRCA1$&$9.434\%$\\
$ssDNARPA$&$9.434\%$\\
$FANCD1N$&$9.434\%$\\
$RAD51$&$9.434\%$\\
$HRR$&$9.434\%$\\
$USP1$&$9.434\%$\\
$CHK2$&$9.434\%$\\
$H2AX$&$9.434\%$\\
$FAcore$&$8.019\%$\\
$FANCD2I$&$8.019\%$\\
$FAN1$&$8.019\%$\\
$p53$&$8.019\%$\\
$CHK1$&$8.019\%$\\
$XPF$&$7.547\%$\\
$ATR$&$2.358\%$\\
$MUS81$&$0.943\%$\\
$PCNATLS$&$0.472\%$\\
$KU$&$0.472\%$\\
$DNAPK$&$0.472\%$\\
$NHEJ$&$0.472\%$\\
$CHKREC$&$0\%$
\end{tabular}
\end{center}
In this case study, DNA damages such as ICLs and DSBs are the pathological events. Unsurprisingly, the algorithm suggests them to be targeted: this is a logical consequence. However, DNA damages are not biomolecules in themselves and directly targeting them by means of drugs appears senseless. What is relevant are the biomolecules of the FA\slash BRCA pathway suggested as therapeutic targets. Interestingly, ATM dominates all the other candidates, predicting it to be a pivotal therapeutic target for the \textit{patho2} condition, namely the FA\slash BRCA pathway bearing FANCD1\slash BRCA2 or FANCN\slash PALB2 gene null mutation, as observed in Fanconi anemia.

\section{Conclusion}
\label{conclusion}
Under the assumption that attractors of dynamical systems and phenotypes of biological networks are linked when the former models the latter, the results show that the algorithm succeeds in performing the proposed \textit{in silico} target identification. It returns therapeutic bullets for a pathological variant of the mammalian cell cycle relevant in cancer and for a pathological variant modeling Fanconi anemia. Consequently, the algorithm can be used on other synchronous Boolean models of biological networks involved in diseases for \textit{in silico} target identification. It is intended to be of use in the early steps of target identification by providing an efficient way to identify candidate targets prior to costly wet lab experiments. However, both the physiological and pathological variants have to be known. This can constitute a limit of the proposed methodology since not all the pathophysiologies are known.

Target identification, whether performed \textit{in silico} or not, is a step belonging to a wider process: drug discovery. Having demonstrated a potential target \textit{in silico}, or even \textit{in vitro}, is far from having a medication. Further work and many years are necessary before obtaining a drug which is effective \textit{in vivo}. For example, and among other characteristics, such a drug has to be absorbed by the organism, has to reach its target and has to be non-toxic at therapeutic dosages. Furthermore, as with any \textit{in silico} evidence, it should be validated through wet lab experiments: there is a bridge to cross between theory and practice. Indeed, mathematical models approximate reality without reproducing it and theory must meet practice. For example, targeting ATM should restore a physiological running of the FA\slash BRCA pathway bearing FANCD1\slash BRCA2 or FANCN\slash PALB2 gene null mutation. However, if ATM operates in other pathways, targeting it may disturb them, thus potentially creating \textit{de novo} non-physiological conditions. Nevertheless, it is expected that the algorithm is of interest for target identification, notably by exploiting the inexpensiveness and predictive power of computational approaches to optimize the efficiency of costly wet lab experiments..

While finding Boolean network attractors of biological networks is not the purpose of this work, it is a necessary step which is in itself a challenging field of computational biology. Therefore, incorporating advances made in this field could be an interesting improvement. Another possible improvement could be to extend the algorithm for asynchronous Boolean networks since such models are likely to more accurately describe the dynamics of biological systems \cite{zhu2014asynchronous,liang2012stochastic}. Indeed, in biological systems, events may be subjected to stochasticity, may not occur simultaneously or may not belong to the same time scale, three points that a synchronous updating scheme does not take into account.

Yet another possible improvement could be to use a finer logic, such as multivalued logic. One of the main limitations of Boolean models is that variables can take only two values. In reality, things are not necessarily binary and variables should be able to take more values. Multivalued logic enables it in a discrete manner where variables can take a finite number of values between $0$ (false) and $1$ (true). For example, one can state that Rb is partly impaired rather than totally. Such a statement is not implementable with Boolean models but is with multivalued ones such as, for example, a three-valued logic where $true=1$, $moderate=0.5$ and $false=0$.

Finally, considering the basin of attraction of the pathological attractors could be an interesting extension of the criterion for selecting therapeutic bullets. In that case, the therapeutic potential of bullets could be assessed by estimating their ability at reducing the basin of the pathological attractors, as performed by Fumia \textit{et al} with their Boolean model of cancer pathways \cite{fumia2013boolean}. Such a criterion enables to consider the particular case where pathological attractors are removed, that is where pathological basins are reduced to the empty set, but also the other cases where pathological basins are not necessarily reduced to the empty set. Such a less restrictive selection of therapeutic bullets would enable to consider more targeting strategies for counteracting diseases.

\section{Additional improvements}
First of all, some additional definitions should be stated:
\begin{itemize}
\item \textbf{physiological state space}: The state space $S_{physio}$ of the physiological variant.
\item \textbf{pathological state space}: The state space $S_{patho}$ of the pathological variant.
\item \textbf{testing state space}: The state space $S_{test}$ of the pathological variant under the effect of a bullet.
\item \textbf{physiological basin}: The basin of attraction $B_{physio,i}$ of a physiological attractor $a_{physio,i}$.
\item \textbf{pathological basin}: The basin of attraction $B_{patho,i}$ of a pathological attractor $a_{patho,i}$.
\item \textbf{$n$-bullet}: A bullet made of $n$ targets.
\end{itemize}
Among the possible improvements mentioned in the \hyperref[conclusion]{\textit{Conclusion}} section page \pageref{conclusion}, two are done: extending the algorithm for multivalued logic and considering pathological basins for selecting therapeutic bullets.

\subsection{Multivalued logic}
\subsubsection{Introduction}
One of the main limitations of Boolean networks is that variables can take only two values, which can be quite simplistic. Depending on what variables model, such as activity level of enzymes or abundance of gene products, considering more than two possible levels should enable models to be more realistic. Without leaving the logic-based modeling formalism, one solution is to extend Boolean logic to multivalued logic \cite{rescher1968many}. As with Boolean logic, variables of multivalued logic are discrete, their value belonging to $[\![0;1]\!]$ where $0$ means false and $1$ means true. With Boolean logic, only $0$ and $1$ can be used to valuate variables. With multivalued logic, an arbitrary finite number $h$ of values in $[\![0;1]\!]$ can be used. Therefore, variables of multivalued logic can model more than only two possible levels, enabling models to be more realistic than those based on Boolean logic.

\subsubsection{Methods}
Boolean logic can be seen as a particular case of multivalued logic: it is a bivalued logic where variables take their value in $\lbrace 0,1\rbrace$. While Boolean operators work well in this case, multivalued logic requires suitable logical operators to be introduced. One solution is to use a mathematical formulation of the Boolean operators which also works with any multivalued logic, just as the Zadeh operators. These logical operators are a mathematical generalization of the Boolean ones proposed for fuzzy logic by its pioneer Lotfi Zadeh. Their mathematical formulation is as follow:

\begin{IEEEeqnarray*}{r C l}
AND(x,y)&=&min(x,y)\\
OR(x,y)&=&max(x,y)\\
NOT(x)&=&1-x
\end{IEEEeqnarray*}

With a $h$-valued logic, $card\ S=h^{n}$. If $h=2$ then this is the Boolean case, where $card\ S$ already raises computational difficulties. With an arbitrary $h>2$, $card\ S$ raises even more computational difficulties. The same applies to the testable bullets since there are $h^{r}$ possible $c_{moda}$ and then $(n!\cdot h^{r})/(r!\cdot (n-r)!)$ possible bullets. To illustrate how the algorithm works with a multivalued logic without overloading it, a $3$-valued logic is used with $\lbrace 0,0.5,1\rbrace$ as domain of value: $x_{i}(k)\in \lbrace 0,0.5,1\rbrace$. $0$ and $1$ have the same meaning as in Boolean logic, namely false and true respectively. $0.5$ is an intermediate truth degree which can be seen as an intermediate level of activity or abundance, depending on what is modeled. Consequently, $S=\lbrace 0,0.5,1\rbrace^{n}$ implying $\boldsymbol{x}_{0},\boldsymbol{x}(k)\in \lbrace 0,0.5,1\rbrace^{n}$, $D\subseteq \lbrace 0,0.5,1\rbrace^{n}$ and $moda_{i}\in \lbrace 0,0.5,1\rbrace$. Moreover, the Boolean operators of the $f_{i}$ are replaced by the Zadeh operators. This results in the following minor changes in the pseudocode of the algorithm described in \hyperref[appendix1]{\textit{Appendix \ref*{appendix1}}} page \pageref{appendix1}:

\begin{center}
\begin{tabular}{c|c|c}
line&Boolean logic&$h$-valued logic\\ \hline
$2$&$card\ D=min(card\ D,2^{n})$&$card\ D=min(card\ D,h^{n})$\\
$29$&$max_{moda}^{r}=min(max_{moda},2^{r})$&$max_{moda}^{r}=min(max_{moda},h^{r})$
\end{tabular}
\end{center}

How the algorithm works with this $3$-valued logic is illustrated with the example network, whose the logical functions become:

\begin{footnotesize}
\begin{IEEEeqnarray*}{r C l}
CycD_{+}&=&CycD\\
Rb_{+}&=&max(min(1-CycD,1-CycE,1-CycA,1-CycB),\\
&&min(p27,1-CycD,1-CycB))\\
E2F_{+}&=&max(min(1-Rb,1-CycA,1-CycB),min(p27,1-Rb,1-CycB))\\
CycE_{+}&=&min(E2F,1-Rb)\\
CycA_{+}&=&max(min(E2F,1-Rb,1-Cdc20,1-min(Cdh1,UbcH10)),\\
&&min(CycA,1-Rb,1-Cdc20,1-min(Cdh1,UbcH10)))\\
p27_{+}&=&max(min(1-CycD,1-CycE,1-CycA,1-CycB),\\
&&min(p27,1-min(CycE,CycA),1-CycB,1-CycD))\\
Cdc20_{+}&=&CycB\\
Cdh1_{+}&=&max(min(1-CycA,1-CycB),Cdc20,min(p27,1-CycB))\\
UbcH10_{+}&=&max(1-Cdh1,min(Cdh1,UbcH10,max(Cdc20,CycA,CycB)))\\
CycB_{+}&=&min(1-Cdc20,1-Cdh1)
\end{IEEEeqnarray*}
\end{footnotesize}
which is $\boldsymbol{f}_{physio}$. For $\boldsymbol{f}_{patho}$, owing to this $3$-valued logic, a constitutive but partial inactivation of Rb is simulated. The corresponding $f_{i}$ becomes:

\begin{IEEEeqnarray*}{c}
Rb_{+}=0.5
\end{IEEEeqnarray*}
in $\boldsymbol{f}_{patho}$.

\subsubsection{Results}
With the example network modeled by this $3$-valued logic, $card\ S=59\ 049$, which remains computationally tractable. Therefore, $card\ D=card\ S$: all the attractors are found. With the physiological variant, the algorithm returns:

\begin{IEEEeqnarray*}{c}
A_{physio}=\lbrace a_{physio1},a_{physio2},a_{physio3},a_{physio4},a_{physio5},a_{physio6}\rbrace
\end{IEEEeqnarray*}
where

\begin{center}
$a_{physio1}=$
\begin{tabular}{l|l}
$CycD$&$0$\\
$Rb$&$0.5$\\
$E2F$&$0.5$\\
$CycE$&$0.5$\\
$CycA$&$0.5$\\
$p27$&$0.5$\\
$Cdc20$&$0.5$\\
$Cdh1$&$0.5$\\
$UbcH10$&$0.5$\\
$CycB$&$0.5$
\end{tabular}\hfill
$a_{physio2}=$
\begin{tabular}{l|l}
$CycD$&$0$\\
$Rb$&$1$\\
$E2F$&$0$\\
$CycE$&$0$\\
$CycA$&$0$\\
$p27$&$1$\\
$Cdc20$&$0$\\
$Cdh1$&$1$\\
$UbcH10$&$0$\\
$CycB$&$0$
\end{tabular}
\end{center}

\begin{center}
$a_{physio3}=$
\begin{tabular}{l|l}
$CycD$&$0.5$\\
$Rb$&$0.5$\\
$E2F$&$0.5$\\
$CycE$&$0.5$\\
$CycA$&$0.5$\\
$p27$&$0.5$\\
$Cdc20$&$0.5$\\
$Cdh1$&$0.5$\\
$UbcH10$&$0.5$\\
$CycB$&$0.5$
\end{tabular}\hfill
$a_{physio4}=$
\begin{tabular}{l|l}
$CycD$&$1$\\
$Rb$&$0$\\
$E2F$&$0.5$\\
$CycE$&$0.5$\\
$CycA$&$0.5$\\
$p27$&$0$\\
$Cdc20$&$0.5$\\
$Cdh1$&$0.5$\\
$UbcH10$&$0.5$\\
$CycB$&$0.5$
\end{tabular}
\end{center}

\begin{center}
$a_{physio5}=$
\begin{tabular}{l|l l}
$CycD$&$0$&$0$\\
$Rb$&$0.5$&$1$\\
$E2F$&$0$&$0.5$\\
$CycE$&$0$&$0$\\
$CycA$&$0$&$0$\\
$p27$&$0.5$&$1$\\
$Cdc20$&$0.5$&$0$\\
$Cdh1$&$0.5$&$1$\\
$UbcH10$&$0.5$&$0.5$\\
$CycB$&$0$&$0.5$
\end{tabular}
\end{center}

\begin{center}
$a_{physio6}=$
\begin{tabular}{l|l l l l l l l}
$CycD$&$1$&$1$&$1$&$1$&$1$&$1$&$1$\\
$Rb$&$0$&$0$&$0$&$0$&$0$&$0$&$0$\\
$E2F$&$0$&$1$&$1$&$1$&$0$&$0$&$0$\\
$CycE$&$0$&$0$&$1$&$1$&$1$&$0$&$0$\\
$CycA$&$0$&$0$&$0$&$1$&$1$&$1$&$1$\\
$p27$&$0$&$0$&$0$&$0$&$0$&$0$&$0$\\
$Cdc20$&$1$&$0$&$0$&$0$&$0$&$0$&$1$\\
$Cdh1$&$1$&$1$&$1$&$1$&$0$&$0$&$0$\\
$UbcH10$&$1$&$1$&$0$&$0$&$0$&$1$&$1$\\
$CycB$&$0$&$0$&$0$&$0$&$0$&$1$&$1$
\end{tabular}
\end{center}
and their corresponding basin of attraction:

\begin{center}
\begin{tabular}{c|c}
$a_{i}$&$B_{i}$ (in \% of $card\ S_{physio}$)\\ \hline
$a_{physio1}$&$9.9\%$\\
$a_{physio2}$&$20.1\%$\\
$a_{physio3}$&$33.3\%$\\
$a_{physio4}$&$24.5\%$\\
$a_{physio5}$&$3.4\%$\\
$a_{physio6}$&$8.8\%$
\end{tabular}
\end{center}
It should be noted that $a_{physio2}$ and $a_{physio6}$ are the two physiological attractors found in the Boolean case. Indeed, since $\lbrace 0,1\rbrace \subset \lbrace 0,0.5,1\rbrace$ and since the Zadeh operators also work with Boolean logic, Boolean logic is included in this three-valued logic. This means that results obtainable with the former are also obtainable with the latter. With the pathological variant, where Rb is constitutively but partially inactivated, the algorithm returns:

\begin{IEEEeqnarray*}{c}
A_{patho}=\lbrace a_{physio1},a_{physio3},a_{patho1}\rbrace
\end{IEEEeqnarray*}
where

\begin{center}
$a_{patho1}=$
\begin{tabular}{l|l}
$CycD$&$1$\\
$Rb$&$0.5$\\
$E2F$&$0.5$\\
$CycE$&$0.5$\\
$CycA$&$0.5$\\
$p27$&$0$\\
$Cdc20$&$0.5$\\
$Cdh1$&$0.5$\\
$UbcH10$&$0.5$\\
$CycB$&$0.5$
\end{tabular}
\end{center}
ans their corresponding basin of attraction:

\begin{center}
\begin{tabular}{c|c}
$a_{i}$&$B_{i}$ (in \% of $card\ S_{patho}$)\\ \hline
$a_{physio1}$&$33.3\%$\\
$a_{physio3}$&$33.3\%$\\
$a_{patho1}$&$33.3\%$
\end{tabular}
\end{center}
Only $a_{physio1}$ and $a_{physio3}$ remain, while $a_{patho1}$ appears and is what therapeutic bullets have to remove from $S_{test}$.

As in the Boolean case, the algorithm is launched with $r_{min}=1$ and $r_{max}=2$. $max_{targ}$ and $max_{moda}$ are set to their maximum, namely $max_{targ}=45$ and $max_{moda}=9$: all the $1,2$-bullets are tested. The algorithm returns the following therapeutic bullets:

\begin{center}
\begin{tabular}{l l|l}
$CycD[0]$&&silver\\
$CycD[0.5]$&&silver\\
$CycD[0]$&$Rb[0.5]$&silver\\
$CycD[0.5]$&$Rb[0.5]$&silver\\
$CycD[1]$&$Rb[0]$&silver\\
$CycD[0]$&$E2F[0.5]$&silver\\
$CycD[0.5]$&$E2F[0.5]$&silver\\
$CycD[0]$&$CycE[0.5]$&silver\\
$CycD[0.5]$&$CycE[0.5]$&silver\\
$CycD[0]$&$CycA[0.5]$&silver\\
$CycD[0.5]$&$CycA[0.5]$&silver\\
$CycD[0]$&$p27[0.5]$&silver\\
$CycD[0.5]$&$p27[0.5]$&silver\\
$CycD[0]$&$Cdc20[0.5]$&silver\\
$CycD[0.5]$&$Cdc20[0.5]$&silver\\
$CycD[0]$&$Cdh1[0.5]$&silver\\
$CycD[0.5]$&$Cdh1[0.5]$&silver\\
$CycD[0]$&$UbcH10[0.5]$&silver\\
$CycD[0.5]$&$UbcH10[0.5]$&silver\\
$CycD[0]$&$CycB[0.5]$&silver\\
$CycD[0.5]$&$CycB[0.5]$&silver
\end{tabular}
\end{center}
where $X[y]$ means that the node $X\in V$ has to be set to the value $y\in \lbrace 0,0.5,1\rbrace$. For example, the third therapeutic bullet is made of the targets CycD and Rb whose the value has to be set to $0$ and $0.5$ respectively. As in the Boolean case, it should be noted that no golden bullets are found, an unsurprising result since they are rarer than silver ones.

\subsubsection{Conclusion}
The algorithm is now extended for multivalued logic, which includes the Boolean one. This means that the previous strictly Boolean version of the algorithm is included in this new one. Moreover, allowing variables to take an arbitrary finite number of values should enable to more accurately model biological processes and produce more fine-tuned therapeutic bullets. However, this accuracy and fine-tuning are at the cost of an increased computational requirement. Indeed, in this work, the computational requirement essentially depends on the cardinality of the state space, which itself depends on the size of the model and the used multivalued logic. Therefore, the size of the model and the used multivalued logic should be balanced: the smaller the model is, the more variables should be finely valued. For example, for a fine therapeutic investigation, the model should only contain the essential and specific pieces of the pathophysiology of interest, modeled by a finely valued logic. On the other hand, for a gross therapeutic investigation, an exhaustive model could be used but modeled by a coarse-grained logic, such as the Boolean one. Finally, it should be noted that the ultimate multivalued logic is the infinitely valued one, which is fuzzy logic. With fuzzy logic, the whole $[0;1] \subset \mathbb{R}$ is used to valuate variables, which should bring the best accuracy for the qualitative modeling formalism \cite{poret2014logic}.

\subsection{Therapeutic bullet assessment}
\subsubsection{Introduction}
Till now, the algorithm requires therapeutic bullets to remove all the pathological attractors from the pathological state space, so that the pathological variant no longer exhibits pathological phenotypes. This criterion for selecting therapeutic bullets can appear somewhat drastic since it is all or nothing. A less strict criterion should enable to consider more targeting strategies, and then more possibilities for counteracting diseases. Certainly, a less restrictive criterion could bring less ``powerful'' therapeutic bullets, but being too demanding potentially leads to no results and loss of nonetheless interesting findings.

The therapeutic potential of bullets could be assessed by estimating their ability at reducing the cardinality of the pathological basins. This is a more permissive criterion since therapeutic bullets no longer have to necessarily remove the pathological attractors. Reducing the cardinality of a pathological basin renders the corresponding pathological attractor less reachable, and then the associated pathological phenotype less likely. This new criterion includes the previous one: removing an attractor means reducing its basin of attraction to the empty set. Therefore, therapeutic bullets obtainable with the previous criterion are also obtainable with this new one.

\subsubsection{Methods}
To implement this new criterion for selecting therapeutic bullets, the algorithm considers a bullet as therapeutic if it increases $card\ \bigcup B_{physio,i}$ in $S_{test}$ without creating \textit{de novo} attractors. Since the attractors are either physiological or pathological, increasing $card\ \bigcup B_{physio,i}$ is equivalent to decreasing $card\ \bigcup B_{patho,i}$. The goal of this new criterion is to increase the physiological part of $S_{test}$, which is equivalent to decreasing its pathological part. Consequently, a pathological variant treated by such a therapeutic bullet tends to, but not necessarily reaches, an overall physiological behavior. However, as with the previous criterion, it does not ensure that the $a_{physio,i}$ are preserved\slash restored. \textit{A fortiori}, it does not ensure that the $B_{physio,i}$ in $S_{test}$ are as in $S_{physio}$. This means that it does not ensure that the reachability of the $a_{physio,i}$ is preserved\slash restored. Nevertheless, as with the previous criterion, this is a matter of choice between a therapeutic bullet or not. To assist this choice and better visualize the effects of therapeutic bullets, the $card\ B_{physio,i}$ and $card\ B_{patho,i}$ in $S_{test}$ are computed.

Implementing this new criterion for selecting therapeutic bullets is a major change. Therefore, the pseudocode of the algorithm presented in \hyperref[appendix1]{\textit{Appendix \ref*{appendix1}}} page \pageref{appendix1} is rewritten and structured into three modules:
\begin{itemize}
\item the $compute\_A$ function, which computes $A_{physio}$ or $A_{patho}$, depending on which of the $\boldsymbol{f}_{physio}$ or $\boldsymbol{f}_{patho}$ is passed
\item the $compute\_cover$ function, which for two attractor sets $A_{1}$ and $A_{2}$ computes the covering of $S_{2}$ by $\bigcup B_{1,i}$, expressed in percents of $card\ S_{2}$
\item the $compute\_T$ function, which computes a set $T$ of therapeutic bullets
\end{itemize}
Below is the corresponding pseudocode:\newline

\noindent \textbf{function} $A=compute\_A(\boldsymbol{f},c_{targ},c_{moda},D,V)$
\begin{algorithmic}[1]
\STATE $A=\lbrace \rbrace$
\FOR{$\boldsymbol{x}_{0}\in D$}
    \STATE $k=1$
    \STATE $\boldsymbol{x}(k)=\boldsymbol{x}_{0}$
    \WHILE{\TRUE}
        \STATE $\boldsymbol{x}(k+1)=\boldsymbol{f}(\boldsymbol{x}(k))$
        \FOR{$targ_{i}\in c_{targ}$}
            \FOR{$v_{j}\in V$}
                \IF{$v_{j}=targ_{i}$}
                    \STATE $x_{j}(k+1)=moda_{i}$
                \ENDIF
            \ENDFOR
        \ENDFOR
        \IF{$\exists k'\in [\![1,k]\!]\colon \boldsymbol{x}(k+1)=\boldsymbol{x}(k')$}
            \STATE $a_{i}.seq=(\boldsymbol{x}(k'),\dotsc,\boldsymbol{x}(k))$
            \IF{$\exists a_{j}\in A\colon a_{i}.seq=a_{j}.seq$}
                \STATE $a_{j}.freq=a_{j}.freq+1$
            \ELSE
                \STATE $a_{i}.freq=1$
                \STATE $A=A\cup \lbrace a_{i}\rbrace$
            \ENDIF
            \STATE \textbf{break}
        \ENDIF
        \STATE $k=k+1$
    \ENDWHILE
\ENDFOR
\FOR{$a\in A$}
    \STATE $a.freq=a.freq\cdot 100/card\ D$
\ENDFOR
\RETURN $A$
\end{algorithmic}
\textbf{end function}\\
For $A_{physio}$ and $A_{patho}$, which are computed without bullet, the empty bullet $((),())$ has to be passed. The $a_{i}$ are represented as structures composed of two fields: $a_{i}.seq$, which is the sequence of $a_{i}$ (line 15), and $a_{i}.freq$, which is the corresponding $card\ B_{i}$, expressed in percents of $card\ D$. To compute $a_{i}.freq$, the algorithm counts the number of times $a_{i}$ is reached (line 19 if this is the first time $a_{i}$ is reached, line 17 otherwise) and then, once all the $x_{0}\in D$ are computed, translates $a_{i}.freq$ in percents of $card\ D$ (line 28).\newline

\noindent \textbf{function} $y=compute\_cover(A_{1},A_{2})$
\begin{algorithmic}[1]
\STATE $cover=0$
\FOR{$a_{1}\in A_{1}$}
    \IF{$\exists a_{2}\in A_{2}\colon a_{1}.seq=a_{2}.seq$}
        \STATE $cover=cover+a_{2}.freq$
    \ENDIF
\ENDFOR
\RETURN $cover$
\end{algorithmic}
\textbf{end function}\\
If $a_{1}$ also belongs to $A_{2}$ (line 3) then the cardinality of its basin in $S_{2}$ is used to compute the covering of $S_{2}$ by $\bigcup B_{1,i}$ (line 4).\newline

\noindent \textbf{function} $T=compute\_T(\boldsymbol{f}_{physio},\boldsymbol{f}_{patho},r_{min},r_{max},max_{targ},max_{moda},\\ max_{D},h,V)$
\begin{algorithmic}[1]
\STATE $n=card\ V$
\STATE $D=\lbrace \rbrace$
\WHILE{$card\ D<max_{D}$}
    \STATE generate $x_{0}\notin D$
    \STATE $D=D\cup \lbrace x_{0}\rbrace$
\ENDWHILE
\STATE $A_{physio}=compute\_A(\boldsymbol{f}_{physio},(),(),D,V)$
\STATE $A_{patho}=compute\_A(\boldsymbol{f}_{patho},(),(),D,V)$
\STATE $T=\lbrace \rbrace$
\STATE $cover_{patho}=compute\_cover(A_{physio},A_{patho})$
\FOR{$r\in [\![r_{min},r_{max}]\!]$}
    \STATE $C_{targ}=\lbrace \rbrace$
    \STATE $C_{moda}=\lbrace \rbrace$
    \WHILE{$card\ C_{targ}<min(max_{targ},n!/(r!\cdot (n-r)!))$}
        \STATE generate $c_{targ}\notin C_{targ}$
        \STATE $C_{targ}=C_{targ}\cup \lbrace c_{targ}\rbrace$
    \ENDWHILE
    \WHILE{$card\ C_{moda}<min(max_{moda},h^{r})$}
        \STATE generate $c_{moda}\notin C_{moda}$
        \STATE $C_{moda}=C_{moda}\cup \lbrace c_{moda}\rbrace$
    \ENDWHILE
    \FOR{$c_{targ}\in C_{targ}$}
        \FOR{$c_{moda}\in C_{moda}$}
            \STATE $A_{test}=compute\_A(\boldsymbol{f}_{patho},c_{targ},c_{moda},D,V)$
            \IF{$A_{test}\subseteq A_{physio}\cup A_{patho}$}
                \STATE $cover_{test}=compute\_cover(A_{physio},A_{test})$
                \IF{$cover_{test}>cover_{patho}$}
                    \STATE $T=T\cup \lbrace (c_{targ},c_{moda})\rbrace$
                \ENDIF
            \ENDIF
        \ENDFOR
    \ENDFOR
\ENDFOR
\RETURN $T$
\end{algorithmic}
\textbf{end function}\\
$max_{D}$ is the desired $card\ D$ and $h$ is the cardinality of the domain of value, which depends on the used multivalued logic. $A_{physio}$ and $A_{patho}$ are computed without bullet, so the empty bullet $((),())$ is passed to $compute\_A$ (lines 7 and 8). $cover_{patho}$ is the covering of $S_{patho}$ by $\bigcup B_{physio,i}$ (line 10) and $cover_{test}$ is the covering of $S_{test}$ by $\bigcup B_{physio,i}$ (line 26). $A_{test}$ is the pathological attractor set under the effect of the tested bullet (line 24). A therapeutic bullet has to avoid the appearance of \textit{de novo} attractors (line 25) and has to increase the covering of $S_{test}$ by $\bigcup B_{physio,i}$ (line 27).

\subsubsection{Results}
This new criterion for selecting therapeutic bullets is illustrated on the case study modeled by Boolean logic: $h=2$. Since \textit{patho1} has the same attractor set than the physiological variant, only \textit{patho2} is computed. As previously, wholly computing $S$ is too demanding. Therefore, $D$ is intended to have a reasonable cardinality: $max_{D}=100\ 000$. All the $1,2$-bullets are tested: $r_{min}=1$, $r_{max}=2$, $max_{targ}=378$ and $max_{moda}=4$. However, their therapeutic potential is no longer expressed as golden or silver but by their gain. It is displayed as follow: $x\% \to y\%$ where $card\ \bigcup B_{physio,i}=x\%$ in $S_{patho}$ and $y\%$ in $S_{test}$. Consequently, in order to increase the physiological part of $S_{test}$, a therapeutic bullet has to make $y>x$. The $card\ B_{physio,i}$ and $card\ B_{patho,i}$ in $S_{test}$ are also computed and expressed in percents of $card\ S_{test}$. The algorithm returns $59$ therapeutic bullets whose the list can be found in \hyperref[appendix3]{\textit{Appendix \ref*{appendix3}}} page \pageref{appendix3}.

A therapeutic bullet as defined by the previous criterion, that is which removes all the $a_{patho,i}$ from $S_{test}$, makes \textit{de facto} $card\ \bigcup B_{physio,i}=100\%$ in $S_{test}$. As already mentioned, the previous criterion is included in this new one: therapeutic bullets obtainable with the former are also obtainable with the latter. This can be checked by noting that the $1,2$-therapeutic bullets found with the previous criterion are also found with this new one.

With this case study, $A_{physio}=\lbrace a_{physio1}\rbrace$, so $\bigcup B_{physio,i}=B_{physio1}$. Therefore, in this particular case where $card\ A_{physio}=1$, therapeutic bullets have to increase $card\ B_{physio1}$ in $S_{test}$. It should be recalled that $card\ B_{physio1}=29.4\%$ in $S_{patho}$, so therapeutic bullets have to make $card\ B_{physio1}>29.4\%$ in $S_{test}$. For example, below are the computed $1$-therapeutic bullets:

\begin{center}
\begin{tabular}{l|r c r|r r}
\multicolumn{1}{c}{bullet}&\multicolumn{3}{|c|}{gain}&$B_{physio1}$&$B_{patho1}$\\ \hline
$-FANCM$&$29.4\%$&$\to$&$44.6\%$&$44.6\%$&$55.4\%$\\
$-FANCD2I$&$29.4\%$&$\to$&$30.4\%$&$30.4\%$&$69.6\%$\\
$-XPF$&$29.4\%$&$\to$&$46.2\%$&$46.2\%$&$53.8\%$\\
$-FAN1$&$29.4\%$&$\to$&$32.9\%$&$32.9\%$&$67.1\%$\\
$-ATM$&$29.4\%$&$\to$&$100\%$&$100\%$&$0\%$
\end{tabular}
\end{center}

$-ATM$ is a therapeutic bullet also found with the previous criterion since it removes all the $a_{patho,i}$, namely $a_{patho1}$, from $S_{test}$. However, the other four therapeutic bullets are only obtainable with this new criterion since they do not remove $a_{patho1}$ from $S_{test}$. Nevertheless, as therapeutic bullets, they increase $card\ B_{physio1}$ in $S_{test}$. This highlight the ability of this new criterion to unravel more therapeutic bullets of varying therapeutic potential, thus opening the way for more targeting strategies of varying theoretical efficacy. Of course, therapeutic bullets of poor potential are also unraveled, such as $-FANCD2I$ which only increases $card\ B_{physio1}$ from $29.4\%$ in $S_{patho}$ to $30.4\%$ in $S_{test}$. However, \textit{in silico} tools should not restrict their predictions to only those exhibiting a high theoretical potency since predicted does not necessarily mean true. Indeed, a prediction of apparently poor interest can reveal itself of great interest in practice, and \textit{vice versa}.

\subsubsection{Conclusion}
The algorithm now uses a new criterion for selecting therapeutic bullets which brings a wider range of targeting strategies intended to push pathological behaviors toward physiological ones with varying predicted efficacy. Moreover, no information is lost from the previous criterion since results obtainable with the previous one are also obtainable with this new one. This new criterion is based on a more permissive assumption stating that reducing the reachability of pathological attractors is therapeutic. For an \textit{in silico} tool such as this algorithm, a more permissive assumption is important since theoretical findings have to outlive the bottleneck separating prediction to reality. With a too strict criterion, the risk of highlighting too few candidate targets or to miss some interesting ones is too hight. Indeed, results predicted \textit{in silico} have to be validated \textit{in vitro} and\slash or \textit{in vivo}. Therefore, requiring only perfect predictions such as therapeutic bullets removing all the pathological attractors could left insufficient results after validation. All the more so that a prediction of apparently poor interest could reveal itself as an insight of great interest and \textit{vice versa}.

This new criterion for selecting therapeutic bullets also brings a finer assessment of their potential since all the percentages between $card\ \bigcup B_{physio,i}$ in $S_{patho}$ and $100\%$ are considered. With the previous criterion, the only therapeutic potential is $card\ \bigcup B_{physio,i}=100\%$ in $S_{test}$, thus reducing the assessment to therapeutic or not. However, things are not necessarily black or white but rather a continuum of gray nuances, so the assessment of therapeutic potentials should be nuanced too, just as it is now.

\newpage

\section{Appendices}
\subsection{Appendix 1}
\label{appendix1}
The algorithm in one block of pseudocode.

\begin{algorithmic}[1]
\STATE \textbf{prompt} $card\ D$
\STATE $card\ D=min(card\ D,2^{n})$
\STATE generate $D\subseteq S$
\STATE $H=\lbrace \rbrace$
\STATE $A_{physio}=\lbrace \rbrace$
\FOR{$x_{0}\in D$}
    \STATE $k=1$
    \STATE $\boldsymbol{x}(k)=x_{0}$
    \WHILE{\TRUE}
        \IF{$\exists w\in H\colon \boldsymbol{x}(k)\in w$}
            \STATE \textbf{break}
        \ENDIF
        \STATE $\boldsymbol{x}(k+1)=\boldsymbol{f}_{physio}(\boldsymbol{x}(k))$
        \IF{$\exists k'\in [\![1,k]\!]\colon \boldsymbol{x}(k+1)=\boldsymbol{x}(k')$}
            \STATE $A_{physio}=A_{physio}\cup \lbrace (\boldsymbol{x}(k'),\dotsc,\boldsymbol{x}(k))\rbrace$
            \STATE \textbf{break}
        \ENDIF
        \STATE $k=k+1$
    \ENDWHILE
    \STATE $H=H\cup \lbrace (\boldsymbol{x}(1),\dotsc,\boldsymbol{x}(k))\rbrace$
\ENDFOR
\RETURN $A_{physio}$
\STATE \textbf{prompt} $r_{min},r_{max},max_{targ},max_{moda}$
\STATE $r_{max}=min(r_{max},n)$
\STATE $golden\_set=\lbrace \rbrace$
\STATE $silver\_set=\lbrace \rbrace$
\FOR{$r\in [\![r_{min},r_{max}]\!]$}
    \STATE $max_{targ}^{r}=min(max_{targ},n!/(r!\cdot (n-r)!))$
    \STATE $max_{moda}^{r}=min(max_{moda},2^{r})$
    \STATE $C_{targ}=\lbrace \rbrace$
    \STATE $C_{moda}=\lbrace \rbrace$
    \WHILE{$card\ C_{targ}<max_{targ}^{r}$}
        \STATE generate $c_{targ}\notin C_{targ}$
        \STATE $C_{targ}=C_{targ}\cup \lbrace c_{targ}\rbrace$
    \ENDWHILE
    \WHILE{$card\ C_{moda}<max_{moda}^{r}$}
        \STATE generate $c_{moda}\notin C_{moda}$
        \STATE $C_{moda}=C_{moda}\cup \lbrace c_{moda}\rbrace$
    \ENDWHILE
    \FOR{$c_{targ}\in C_{targ}$}
        \FOR{$c_{moda}\in C_{moda}$}
            \STATE $H=\lbrace \rbrace$
            \STATE $A_{patho}=\lbrace \rbrace$
            \FOR{$x_{0}\in D$}
                \STATE $k=1$
                \STATE $\boldsymbol{x}(k)=x_{0}$
                \WHILE{\TRUE}
                    \IF{$\exists w\in H\colon \boldsymbol{x}(k)\in w$}
                        \STATE \textbf{break}
                    \ENDIF
                    \STATE $\boldsymbol{x}(k+1)=\boldsymbol{f}_{patho}(\boldsymbol{x}(k))$
                    \FOR{$targ_{i}\in c_{targ}$}
                        \FOR{$v_{j}\in V$}
                            \IF{$v_{j}=targ_{i}$}
                                \STATE $x_{j}(k+1)=moda_{i}$
                            \ENDIF
                        \ENDFOR
                    \ENDFOR
                    \IF{$\exists k'\in [\![1,k]\!]\colon \boldsymbol{x}(k+1)=\boldsymbol{x}(k')$}
                        \STATE $A_{patho}=A_{patho}\cup \lbrace (\boldsymbol{x}(k'),\dotsc,\boldsymbol{x}(k))\rbrace$
                        \STATE \textbf{break}
                    \ENDIF
                    \STATE $k=k+1$
                \ENDWHILE
                \STATE $H=H\cup \lbrace (\boldsymbol{x}(1),\dotsc,\boldsymbol{x}(k))\rbrace$
            \ENDFOR
            \IF{$A_{patho}\subseteq A_{physio}$}
                \IF{$A_{patho}=A_{physio}$}
                    \STATE $golden\_set=golden\_set\cup \lbrace (c_{targ},c_{moda})\rbrace$
                \ELSE
                    \STATE $silver\_set=silver\_set\cup \lbrace (c_{targ},c_{moda})\rbrace$
                \ENDIF
            \ENDIF
        \ENDFOR
    \ENDFOR
\ENDFOR
\RETURN $golden\_set,silver\_set$
\end{algorithmic}

\newpage

\subsection{Appendix 2}
\label{appendix2}
Therapeutic bullets found for the case study.

\begin{footnotesize}
\begin{center}
\begin{tabular}{l l l|l}
$-ATM$&&&golden\\
$-ATM$&$-CHK2$&&golden\\
$-HRR$&$-ATM$&&golden\\
$-ssDNARPA$&$-ATM$&&golden\\
$-BRCA1$&$-ATM$&&golden\\
$-MRN$&$-ATM$&&golden\\
$-FAN1$&$-ATM$&&golden\\
$-ICL$&$-DSB$&&golden\\
$-FAcore$&$-ATM$&&golden\\
$-USP1$&$-ATM$&&golden\\
$-ATM$&$-H2AX$&&golden\\
$-ADD$&$-ATM$&&golden\\
$-RAD51$&$-ATM$&&golden\\
$-XPF$&$-ATM$&&golden\\
$-FANCM$&$-ATM$&&golden\\
$-FANCD1N$&$-ATM$&&golden\\
$-ATM$&$-CHK1$&&golden\\
$-ICL$&$-ATM$&&golden\\
$-ATM$&$-p53$&&golden\\
$-FANCJBRCA1$&$-ATM$&&golden\\
$-FANCD2I$&$-ATM$&&golden\\
$-ICL$&$-FANCD1N$&$-ATM$&golden\\
$-ICL$&$-FAcore$&$-DSB$&golden\\
$-BRCA1$&$-USP1$&$-ATM$&golden\\
$-BRCA1$&$-ssDNARPA$&$-ATM$&golden\\
$-BRCA1$&$-ATM$&$-CHK1$&golden\\
$-ADD$&$-ATM$&$-H2AX$&golden\\
$-FAN1$&$-MRN$&$-ATM$&golden\\
$-ATM$&$-CHK2$&$-H2AX$&golden\\
$-ICL$&$-DSB$&$-MRN$&golden\\
$-XPF$&$-MRN$&$-ATM$&golden\\
$-FAcore$&$-FANCD2I$&$-ATM$&golden\\
$-FANCM$&$-ATM$&$-CHK2$&golden\\
$-RAD51$&$-ATM$&$-p53$&golden\\
$-ICL$&$-ssDNARPA$&$-ATM$&golden\\
$-FANCM$&$-ATR$&$-ATM$&golden\\
$-RAD51$&$-ATM$&$-H2AX$&golden\\
$-ADD$&$-FANCD1N$&$-ATM$&golden\\
$-ICL$&$-USP1$&$-ATM$&golden\\
$-FANCM$&$-MRN$&$-ATR$&golden\\
$-MRN$&$-USP1$&$-ATM$&golden\\
$-FAN1$&$-HRR$&$-ATM$&golden\\
$-BRCA1$&$-ATM$&$-H2AX$&golden\\
$-FANCJBRCA1$&$-ADD$&$-ATM$&golden\\
$-MRN$&$-ssDNARPA$&$-ATM$&golden\\
$-FAcore$&$-ssDNARPA$&$-ATM$&golden\\
$-FAcore$&$-FANCD1N$&$-ATM$&golden\\
$-FANCD2I$&$-BRCA1$&$-ATM$&golden\\
$-ADD$&$-MRN$&$-ATM$&golden\\
$-ATM$&$-p53$&$-CHK2$&golden\\
$-RAD51$&$-ATM$&$-CHK2$&golden\\
$-FANCM$&$-ATM$&$-H2AX$&golden\\
$-ADD$&$-PCNATLS$&$-ATM$&golden\\
$-FANCJBRCA1$&$-ATM$&$-p53$&golden\\
$-FANCM$&$-MRN$&$-ATM$&golden
\end{tabular}
\end{center}
\end{footnotesize}

\newpage

\begin{footnotesize}
\begin{center}
\begin{tabular}{l l l|l}
$-FANCJBRCA1$&$-ATM$&$-CHK2$&golden\\
$-FANCD2I$&$-USP1$&$-ATM$&golden\\
$-ADD$&$-ATM$&$-CHK2$&golden\\
$-FANCD2I$&$-FANCD1N$&$-ATM$&golden\\
$-MRN$&$-HRR$&$-ATM$&golden\\
$-ICL$&$-DSB$&$-USP1$&golden\\
$-FAN1$&$-FANCD1N$&$-ATM$&golden\\
$-FAN1$&$-ATM$&$-H2AX$&golden\\
$-FANCJBRCA1$&$-FAN1$&$-ATM$&golden\\
$-ssDNARPA$&$-ATM$&$-H2AX$&golden\\
$-ATM$&$-CHK1$&$-CHK2$&golden\\
$-ADD$&$-HRR$&$-ATM$&golden\\
$-ATM$&$-p53$&$-CHK1$&golden\\
$-FAcore$&$-ATM$&$-H2AX$&golden\\
$-FANCD2I$&$-ATM$&$-CHK2$&golden\\
$-FAN1$&$-RAD51$&$-ATM$&golden\\
$-FANCD2I$&$-RAD51$&$-ATM$&golden\\
$-FANCJBRCA1$&$-XPF$&$-ATM$&golden\\
$-ICL$&$-FANCJBRCA1$&$-DSB$&golden\\
$-ssDNARPA$&$-HRR$&$-ATM$&golden\\
$-MRN$&$-BRCA1$&$-ATM$&golden\\
$-FANCM$&$-FAN1$&$-ATM$&golden\\
$-ssDNARPA$&$-ATM$&$-p53$&golden\\
$-FAN1$&$-ATM$&$-CHK2$&golden\\
$-FANCD2I$&$-ssDNARPA$&$-ATM$&golden\\
$-FANCD2I$&$-FAN1$&$-ATM$&golden\\
$-XPF$&$-HRR$&$-ATM$&golden\\
$-FAN1$&$-BRCA1$&$-ATM$&golden\\
$-ADD$&$-ATM$&$-CHK1$&golden\\
$-FAcore$&$-HRR$&$-ATM$&golden\\
$-XPF$&$-ATM$&$-CHK1$&golden\\
$-ADD$&$-BRCA1$&$-ATM$&golden\\
$-ICL$&$-FAN1$&$-DSB$&golden\\
$-ADD$&$-ATM$&$-p53$&golden\\
$-ICL$&$-MUS81$&$-ATM$&golden\\
$-FAcore$&$-RAD51$&$-ATM$&golden\\
$-ATM$&$-CHK1$&$-H2AX$&golden\\
$-ICL$&$-MRN$&$-ATM$&golden\\
$-ssDNARPA$&$-ATM$&$-CHK2$&golden\\
$-XPF$&$-RAD51$&$-ATM$&golden\\
$-FANCM$&$-ATM$&$-CHK1$&golden\\
$-ICL$&$-DSB$&$-KU$&golden\\
$-ICL$&$-MRN$&$-ATR$&golden\\
$-ssDNARPA$&$-RAD51$&$-ATM$&golden\\
$-FANCJBRCA1$&$-ssDNARPA$&$-ATM$&golden\\
$-XPF$&$-ATM$&$-p53$&golden\\
$-FAcore$&$-MRN$&$-ATM$&golden\\
$-HRR$&$-ATM$&$-H2AX$&golden\\
$-HRR$&$-ATM$&$-p53$&golden\\
$-FANCJBRCA1$&$-FANCD1N$&$-ATM$&golden\\
$-FANCM$&$-ADD$&$-ATM$&golden\\
$-FAcore$&$-ATM$&$-CHK2$&golden\\
$-ICL$&$-ATM$&$-CHK1$&golden\\
$-MRN$&$-FANCD1N$&$-ATM$&golden\\
$-ADD$&$-ssDNARPA$&$-ATM$&golden\\
$-MRN$&$-RAD51$&$-ATM$&golden\\
$-FANCD1N$&$-ATM$&$-p53$&golden\\
$-FANCD1N$&$-RAD51$&$-ATM$&golden\\
$-BRCA1$&$-ATM$&$-CHK2$&golden\\
$-ADD$&$-RAD51$&$-ATM$&golden
\end{tabular}
\end{center}
\end{footnotesize}

\newpage

\begin{footnotesize}
\begin{center}
\begin{tabular}{l l l|l}
$-ICL$&$-DSB$&$-FANCD1N$&golden\\
$-ICL$&$-RAD51$&$-ATM$&golden\\
$-ICL$&$-ATM$&$-CHK2$&golden\\
$-FANCD1N$&$-ATM$&$-H2AX$&golden\\
$-MRN$&$-ATM$&$-H2AX$&golden\\
$-FAcore$&$-FAN1$&$-ATM$&golden\\
$-ICL$&$-XPF$&$-ATM$&golden\\
$-FANCD2I$&$-ADD$&$-ATM$&golden\\
$-FANCD2I$&$-ATM$&$-H2AX$&golden\\
$-ICL$&$-ATR$&$-ATM$&golden\\
$-FANCM$&$-HRR$&$-ATM$&golden\\
$-USP1$&$-ATM$&$-H2AX$&golden\\
$-ICL$&$-DSB$&$-RAD51$&golden\\
$-ICL$&$-ATM$&$-H2AX$&golden\\
$-FANCD1N$&$-USP1$&$-ATM$&golden\\
$-FANCM$&$-FANCD2I$&$-ATM$&golden\\
$-FANCD2I$&$-MRN$&$-ATM$&golden\\
$-FAcore$&$-ADD$&$-ATM$&golden\\
$-ICL$&$-FAcore$&$-ATM$&golden\\
$-FANCM$&$-ssDNARPA$&$-ATM$&golden\\
$-XPF$&$-ATM$&$-H2AX$&golden\\
$-FAcore$&$-USP1$&$-ATM$&golden\\
$-HRR$&$-ATM$&$-CHK1$&golden\\
$-BRCA1$&$-RAD51$&$-ATM$&golden\\
$-FAN1$&$-ADD$&$-ATM$&golden\\
$-FANCJBRCA1$&$-MRN$&$-ATM$&golden\\
$-FANCM$&$-USP1$&$-ATM$&golden\\
$-FANCJBRCA1$&$-ATM$&$-H2AX$&golden\\
$-FANCM$&$-FAcore$&$-ATM$&golden\\
$-HRR$&$-USP1$&$-ATM$&golden\\
$-ICL$&$-FANCM$&$-ATM$&golden\\
$-ICL$&$-DSB$&$-ssDNARPA$&golden\\
$-FAN1$&$-USP1$&$-ATM$&golden\\
$-FANCM$&$-FANCJBRCA1$&$-ATM$&golden\\
$-ssDNARPA$&$-ATM$&$-CHK1$&golden\\
$-FAcore$&$-FANCJBRCA1$&$-ATM$&golden\\
$-FANCD2I$&$-HRR$&$-ATM$&golden\\
$-FANCD2I$&$-FANCJBRCA1$&$-ATM$&golden\\
$-XPF$&$-ssDNARPA$&$-ATM$&golden\\
$-USP1$&$-ATM$&$-CHK1$&golden\\
$-ICL$&$-DSB$&$-ATM$&golden\\
$-ICL$&$-ADD$&$-DSB$&golden\\
$-USP1$&$-ATM$&$-CHK2$&golden\\
$-XPF$&$-BRCA1$&$-ATM$&golden\\
$-RAD51$&$-ATM$&$-CHK1$&golden\\
$-FANCD1N$&$-ATM$&$-CHK2$&golden\\
$-RAD51$&$-HRR$&$-ATM$&golden\\
$-ICL$&$-ATM$&$-p53$&golden\\
$-ICL$&$-DSB$&$-DNAPK$&golden\\
$-FANCM$&$-FANCD1N$&$-ATM$&golden\\
$-BRCA1$&$-FANCD1N$&$-ATM$&golden\\
$-ICL$&$-HRR$&$-ATM$&golden\\
$-FANCJBRCA1$&$-HRR$&$-ATM$&golden\\
$-USP1$&$-ATM$&$-p53$&golden\\
$-XPF$&$-ATM$&$-CHK2$&golden\\
$-ICL$&$-DSB$&$-CHK2$&golden\\
$-ICL$&$-XPF$&$-DSB$&golden\\
$-ssDNARPA$&$-FANCD1N$&$-ATM$&golden\\
$-FANCJBRCA1$&$-RAD51$&$-ATM$&golden\\
$-ICL$&$-DSB$&$-ATR$&golden
\end{tabular}
\end{center}
\end{footnotesize}

\newpage

\begin{footnotesize}
\begin{center}
\begin{tabular}{l l l|l}
$-HRR$&$-ATM$&$-CHK2$&golden\\
$-ADD$&$-USP1$&$-ATM$&golden\\
$-FANCM$&$-RAD51$&$-ATM$&golden\\
$-FANCJBRCA1$&$-ATM$&$-CHK1$&golden\\
$-FANCM$&$-ATM$&$-p53$&golden\\
$-XPF$&$-FANCD1N$&$-ATM$&golden\\
$-FAcore$&$-BRCA1$&$-ATM$&golden\\
$-ICL$&$-DSB$&$-NHEJ$&golden\\
$-BRCA1$&$-ATM$&$-p53$&golden\\
$-BRCA1$&$-HRR$&$-ATM$&golden\\
$-FANCJBRCA1$&$-USP1$&$-ATM$&golden\\
$-ssDNARPA$&$-USP1$&$-ATM$&golden\\
$-ICL$&$-DSB$&$-H2AX$&golden\\
$-FANCM$&$-BRCA1$&$-ATM$&golden\\
$-MRN$&$-ATM$&$-CHK1$&golden\\
$-ICL$&$-FANCJBRCA1$&$-ATM$&golden\\
$-FANCD1N$&$-ATM$&$-CHK1$&golden\\
$-ICL$&$-DSB$&$-BRCA1$&golden\\
$-MRN$&$-ATM$&$-CHK2$&golden\\
$-FANCJBRCA1$&$-BRCA1$&$-ATM$&golden\\
$-FAN1$&$-ssDNARPA$&$-ATM$&golden\\
$-MRN$&$-ATM$&$-p53$&golden\\
$-FANCD1N$&$-HRR$&$-ATM$&golden\\
$-ICL$&$-MUS81$&$-DSB$&golden\\
$-ICL$&$-DSB$&$-p53$&golden\\
$-XPF$&$-USP1$&$-ATM$&golden\\
$-XPF$&$-ADD$&$-ATM$&golden\\
$-ATM$&$-p53$&$-H2AX$&golden\\
$-ICL$&$-FANCM$&$-DSB$&golden\\
$-ICL$&$-DSB$&$-HRR$&golden\\
$-ICL$&$-BRCA1$&$-ATM$&golden\\
$-RAD51$&$-USP1$&$-ATM$&golden\\
$-ICL$&$-FAN1$&$-ATM$&golden\\
$-ICL$&$-ADD$&$-ATM$&golden\\
$-ICL$&$-DSB$&$-CHK1$&golden\\
$-ICL$&$-FANCD2I$&$-DSB$&golden\\
$-ICL$&$-FANCD2I$&$-ATM$&golden
\end{tabular}
\end{center}
\end{footnotesize}

\newpage

\subsection{Appendix 3}
\label{appendix3}
Therapeutic bullets found for the case study using the new criterion.

\begin{footnotesize}
\begin{center}
\begin{tabular}{l l|r c r|r r}
\multicolumn{2}{c}{bullet}&\multicolumn{3}{|c|}{gain}&$B_{physio1}$&$B_{patho1}$\\ \hline
$-FANCM$&&$29.4\%$&$\to$&$44.6\%$&$44.6\%$&$55.4\%$\\
$-FANCD2I$&&$29.4\%$&$\to$&$30.4\%$&$30.4\%$&$69.6\%$\\
$-XPF$&&$29.4\%$&$\to$&$46.2\%$&$46.2\%$&$53.8\%$\\
$-FAN1$&&$29.4\%$&$\to$&$32.9\%$&$32.9\%$&$67.1\%$\\
$-ATM$&&$29.4\%$&$\to$&$100\%$&$100\%$&$0\%$\\
$-ICL$&$-FANCD2I$&$29.4\%$&$\to$&$30.9\%$&$30.9\%$&$69.1\%$\\
$-ICL$&$-MUS81$&$29.4\%$&$\to$&$53\%$&$53\%$&$47\%$\\
$-ICL$&$-XPF$&$29.4\%$&$\to$&$58.6\%$&$58.6\%$&$41.4\%$\\
$-ICL$&$-FAN1$&$29.4\%$&$\to$&$33.9\%$&$33.9\%$&$66.1\%$\\
$-ICL$&$-DSB$&$29.4\%$&$\to$&$100\%$&$100\%$&$0\%$\\
$-ICL$&$-ATM$&$29.4\%$&$\to$&$100\%$&$100\%$&$0\%$\\
$-FANCM$&$-FAcore$&$29.4\%$&$\to$&$45.8\%$&$45.8\%$&$54.2\%$\\
$-FANCM$&$-FANCD2I$&$29.4\%$&$\to$&$46.3\%$&$46.3\%$&$53.7\%$\\
$-FANCM$&$-FAN1$&$29.4\%$&$\to$&$47.3\%$&$47.3\%$&$52.7\%$\\
$-FANCM$&$-ADD$&$29.4\%$&$\to$&$47.3\%$&$47.3\%$&$52.7\%$\\
$-FANCM$&$-FANCD1N$&$29.4\%$&$\to$&$44.6\%$&$44.6\%$&$55.4\%$\\
$-FANCM$&$-RAD51$&$29.4\%$&$\to$&$44.6\%$&$44.6\%$&$55.4\%$\\
$-FANCM$&$-HRR$&$29.4\%$&$\to$&$44.1\%$&$44.1\%$&$55.9\%$\\
$-FANCM$&$-USP1$&$29.4\%$&$\to$&$44.3\%$&$44.3\%$&$55.7\%$\\
$-FANCM$&$-ATM$&$29.4\%$&$\to$&$100\%$&$100\%$&$0\%$\\
$-FAcore$&$-FANCD2I$&$29.4\%$&$\to$&$30.4\%$&$30.4\%$&$69.6\%$\\
$-FAcore$&$-FAN1$&$29.4\%$&$\to$&$33\%$&$33\%$&$67\%$\\
$-FAcore$&$-ATM$&$29.4\%$&$\to$&$100\%$&$100\%$&$0\%$\\
$-FANCD2I$&$-FAN1$&$29.4\%$&$\to$&$33.2\%$&$33.2\%$&$66.8\%$\\
$-FANCD2I$&$-ADD$&$29.4\%$&$\to$&$30.5\%$&$30.5\%$&$69.5\%$\\
$-FANCD2I$&$-FANCD1N$&$29.4\%$&$\to$&$30.4\%$&$30.4\%$&$69.6\%$\\
$-FANCD2I$&$-RAD51$&$29.4\%$&$\to$&$30.4\%$&$30.4\%$&$69.6\%$\\
$-FANCD2I$&$-USP1$&$29.4\%$&$\to$&$30.4\%$&$30.4\%$&$69.6\%$\\
$-FANCD2I$&$-ATM$&$29.4\%$&$\to$&$100\%$&$100\%$&$0\%$\\
$-FANCJBRCA1$&$-ATM$&$29.4\%$&$\to$&$100\%$&$100\%$&$0\%$\\
$-XPF$&$-ADD$&$29.4\%$&$\to$&$46.2\%$&$46.2\%$&$53.8\%$\\
$-XPF$&$-FANCD1N$&$29.4\%$&$\to$&$46.2\%$&$46.2\%$&$53.8\%$\\
$-XPF$&$-RAD51$&$29.4\%$&$\to$&$46.2\%$&$46.2\%$&$53.8\%$\\
$-XPF$&$-HRR$&$29.4\%$&$\to$&$45.3\%$&$45.3\%$&$54.7\%$\\
$-XPF$&$-USP1$&$29.4\%$&$\to$&$46.2\%$&$46.2\%$&$53.8\%$\\
$-XPF$&$-KU$&$29.4\%$&$\to$&$46.1\%$&$46.1\%$&$53.9\%$\\
$-XPF$&$-DNAPK$&$29.4\%$&$\to$&$46.1\%$&$46.1\%$&$53.9\%$\\
$-XPF$&$-NHEJ$&$29.4\%$&$\to$&$41.6\%$&$41.6\%$&$58.4\%$\\
$-XPF$&$-ATM$&$29.4\%$&$\to$&$100\%$&$100\%$&$0\%$\\
$-FAN1$&$-ADD$&$29.4\%$&$\to$&$32.9\%$&$32.9\%$&$67.1\%$\\
$-FAN1$&$-FANCD1N$&$29.4\%$&$\to$&$32.9\%$&$32.9\%$&$67.1\%$\\
$-FAN1$&$-RAD51$&$29.4\%$&$\to$&$32.9\%$&$32.9\%$&$67.1\%$\\
$-FAN1$&$-HRR$&$29.4\%$&$\to$&$32.2\%$&$32.2\%$&$67.8\%$\\
$-FAN1$&$-USP1$&$29.4\%$&$\to$&$32.9\%$&$32.9\%$&$67.1\%$\\
$-FAN1$&$-KU$&$29.4\%$&$\to$&$31.7\%$&$31.7\%$&$68.2\%$\\
$-FAN1$&$-DNAPK$&$29.4\%$&$\to$&$31\%$&$31\%$&$69\%$\\
$-FAN1$&$-ATM$&$29.4\%$&$\to$&$100\%$&$100\%$&$0\%$\\
$-ADD$&$-ATM$&$29.4\%$&$\to$&$100\%$&$100\%$&$0\%$\\
$-MRN$&$-ATM$&$29.4\%$&$\to$&$100\%$&$100\%$&$0\%$\\
$-BRCA1$&$-ATM$&$29.4\%$&$\to$&$100\%$&$100\%$&$0\%$\\
$-ssDNARPA$&$-ATM$&$29.4\%$&$\to$&$100\%$&$100\%$&$0\%$\\
$-FANCD1N$&$-ATM$&$29.4\%$&$\to$&$100\%$&$100\%$&$0\%$\\
$-RAD51$&$-ATM$&$29.4\%$&$\to$&$100\%$&$100\%$&$0\%$\\
$-HRR$&$-ATM$&$29.4\%$&$\to$&$100\%$&$100\%$&$0\%$\\
$-USP1$&$-ATM$&$29.4\%$&$\to$&$100\%$&$100\%$&$0\%$
\end{tabular}
\end{center}
\end{footnotesize}

\newpage

\begin{footnotesize}
\begin{center}
\begin{tabular}{l l|r c r|r r}
\multicolumn{2}{c}{bullet}&\multicolumn{3}{|c|}{gain}&$B_{physio1}$&$B_{patho1}$\\ \hline
$-ATM$&$-p53$&$29.4\%$&$\to$&$100\%$&$100\%$&$0\%$\\
$-ATM$&$-CHK1$&$29.4\%$&$\to$&$100\%$&$100\%$&$0\%$\\
$-ATM$&$-CHK2$&$29.4\%$&$\to$&$100\%$&$100\%$&$0\%$\\
$-ATM$&$-H2AX$&$29.4\%$&$\to$&$100\%$&$100\%$&$0\%$
\end{tabular}
\end{center}
\end{footnotesize}

\newpage

\bibliography{an_in_silico_target_identification}
\bibliographystyle{unsrt}

\end{document}